\documentclass[a4paper,11pt,reqno]{amsart}

\usepackage[margin=1in,includehead,includefoot]{geometry}
\usepackage[utf8]{inputenc}
\usepackage[T1]{fontenc}
\usepackage{lmodern,microtype}
\usepackage{upref}
\usepackage{mathtools,amssymb,amsthm}
\usepackage[foot]{amsaddr}
\usepackage{enumitem}
\usepackage{xifthen}
\usepackage{todonotes}
\usepackage{tikz}
\usepackage{graphicx}
\usepackage{wrapfig}
\usepackage{hyperref}
\usepackage[b]{esvect} 
\usepackage{textcomp}
\usepackage{mdframed}
\usepackage{xpatch}
\usepackage{xcolor}
\usepackage{bookmark}
\usepackage{gensymb}
\graphicspath{{./graphics/}}

\makeatletter
\let\old@setaddresses\@setaddresses
\def\@setaddresses{\bigskip{\parindent 0pt\let\scshape\relax\let\ttfamily\relax\old@setaddresses}}
\makeatother

\newtheorem{thm}{Theorem}

\newtheorem{proposition}[thm]{Proposition}
\newtheorem{lem}[thm]{Lemma}
\theoremstyle{remark}
\newtheorem*{claim}{Claim}

\theoremstyle{definition}
\newtheorem*{defn}{Definition}
\newenvironment{claimproof}[1][\myproofname]{\begin{proof}[#1]}{\end{proof}}


\global\long\def\RR{\mathcal{R}}%
\global\long\def\R{\mathbb{R}}%
\global\long\def\N{\mathbb{N}}%
\global\long\def\Z{\mathbb{Z}}%
\global\long\def\P{\mathcal{P}}%
\global\long\def\OPT{\mathrm{OPT}}%
\global\long\def\Pe{\mathcal{P}_{E}}%
\global\long\def\Ph{\mathcal{P}_{H}}%
\global\long\def\conv{\mathrm{conv}}%
\global\long\def\fpoly{O(1)}%
\global\long\def\area{\mathrm{area}}%
\global\long\def\bPe{\mathcal{\bar{P}}_{E}}%
\global\long\def\bPh{\bar{\mathcal{P}}_{H}}%
\global\long\def\bPm{\bar{\mathcal{P}}_{M}}%
\global\long\def\shr{\mathrm{shr}}%
\global\long\def\bPm{\bar{\mathcal{P}}_{M}}%
\global\long\def\rot{\mathrm{rot}}%
\global\long\def\L{\mathcal{L}}%
\global\long\def\V{\mathcal{V}}%
\global\long\def\S{\mathcal{S}}%
\global\long\def\Pm{\mathcal{P}_{M}}%
\global\long\def\C{\mathcal{C}}%
\global\long\def\F{\mathcal{F}}%
\global\long\def\diam{\mathrm{diam}}%
\global\long\def\OPTef{\mathrm{OPT}_{\mathrm{EF}}}%
\global\long\def\OPTcf{\mathrm{OPT}_{\mathrm{CF}}}%
\global\long\def\dir{\mathrm{dir}}%
\global\long\def\lleft{\mathrm{left}}%
\global\long\def\mmid{\mathrm{mid}}%

\hypersetup{
  pdftitle={On the two-dimensional knapsack problem for convex polygons}
  pdfauthor={Arturo Merino, Andreas Wiese}
}

\title{On the two-dimensional knapsack problem for convex polygons}

\author{Arturo Merino}
\address[Arturo Merino]{Technische Universität Berlin, Germany,}
\email{merino@math.tu-berlin.de}

\author{Andreas Wiese}
\address[Andreas Wiese]{Universidad de Chile, Chile,}
\email{awiese@dii.uchile.cl}

\thanks{Andreas Wiese: partially supported by FONDECYT Regular grant 1170223. Arturo Merino: partially supported by DFG Project 413902284 and ANID Becas Chile 2019-72200522.}

\begin{document}

\begin{abstract}
We study the two-dimensional geometric knapsack problem for convex polygons. Given a set of weighted convex polygons and a square knapsack, the goal is to select the most profitable subset of the given polygons that fits non-overlappingly into the knapsack. We allow to rotate the polygons by arbitrary angles. We present a quasi-polynomial time $O(1)$-approximation algorithm for the general case and a polynomial time $O(1)$-approximation algorithm if all input polygons are triangles, both assuming polynomially bounded integral input data. Also, we give a quasi-polynomial time algorithm that computes a solution of optimal weight under resource augmentation, i.e., we allow to increase the size of the knapsack by a factor of $1+\delta$ for some $\delta>0$ but compare ourselves with the optimal solution for the original knapsack. To the best of our knowledge, these are the first results for two-dimensional geometric knapsack in which the input objects are more general than axis-parallel rectangles or circles and in which the input polygons can be rotated by arbitrary angles.
\end{abstract}

\maketitle

\section{Introduction}

In the two-dimensional geometric knapsack problem (2DKP) we are given
a square knapsack $K:=[0,N]\times[0,N]$ for some integer $N$ and
a set of $n$ convex polygons $\P$ where each polygon $P_{i}\in\P$
has a weight $w_{i}>0$; we write $w(\P'):=\sum_{P_i\in \P'}w_i$ for any set $\P'\subseteq\P$. 
The goal is to select a subset $\P'\subseteq\P$
of maximum total weight $w(\P')$ such that the polygons in $\P'$ fit non-overlapping
into $K$ if we translate and rotate them suitably (by arbitrary angles).
2DKP is a natural packing problem, the reader may think of cutting
items out of a piece of raw material like metal or wood, cutting cookings
out of dough, or, in three dimensions, loading cargo into a ship or
a truck. In particular, in these applications the respective items
can have various kinds of shapes. Also note that 2DKP is a natural
geometric generalization of the classical one-dimensional knapsack
problem.

Our understanding of 2DKP highly depends on the type of input objects.
If all polygons are axis-parallel squares there is a $(1+\epsilon)$-approximation
with a running time of the form $O_\epsilon(1) n^{O(1)}$ (i.e., an 
EPTAS)~\cite{HeydrichWiese2019},
and there can be no FPTAS (unless $\mathsf{P=NP}$) since the problem
is strongly $\mathsf{NP}$-hard~\cite{leung1990packing}. For axis-parallel
rectangles there is a polynomial time $(17/9+\epsilon)<1.89$-approximation
algorithm and a $(3/2+\epsilon)$-approximation if the items can be
rotated by exactly 90 degrees~\cite{2DKPrectangles2017}. If the
input data is quasi-polynomially bounded there is even a $(1+\epsilon)$-approximation
in quasi-polynomial time~\cite{adamaszek2015knapsack}, with and
without the possibility to rotate items by 90 degrees. For circles
a $(1+\epsilon)$-approximation is known under resource augmentation
in one dimension if the weight of each circle equals its area~\cite{lintzmayer2018two}.

To the best of our knowledge, there is no result known for 2DKP for
shapes different than axis-parallel rectangles and circles. Also,
there is no result known in which input polygons are allowed to be
rotated by angles different than 90 degrees. However, in the applications
of 2DKP the items might have shapes that are more complicated than
rectangles or circles. Also, it makes sense to allow rotations by
arbitrary angles, e.g., when cutting items out of some material. In
this paper, we present the first results for 2DKP in these settings.

\subsection{Our contribution}

We study 2DKP for arbitrary convex polygons, allowing to rotate them
by arbitrary angles. Note that due to the latter, it might be that
some optimal solution places the vertices of the polygons on irrational
coordinates, even if all input numbers are integers. Our first results
are a quasi-polynomial time $O(1)$-approximation algorithm for general
convex polygons and a polynomial time $O(1)$-approximation algorithm
for triangles.

By rotation we can assume for each input polygon that the line segment
defining its diameter is horizontal. We identify three different types
of polygons for which we employ different strategies for packing them,
see Figure~\ref{fig:packings}a). First, we consider the \emph{easy
}polygons which are the polygons whose bounding boxes fit into the
knapsack without rotation. We pack these polygons such that their
bounding boxes do not intersect. Using area arguments and the Steinberg's algorithm~\cite{steinberg1997strip} we obtain
a $O(1)$-approximation for the easy polygons. Then we consider the
\emph{medium} polygons which are the polygons whose bounding boxes
easily fit into the knapsack if we can rotate them by 45 degrees.
We use a special type of packing in which the bounding boxes are rotated
by 45 degrees and then stacked on top of each other, see Figure \ref{fig:packings}b).
More precisely, we group the polygons by
the widths of their bounding boxes and to each group we assign two
rectangular containers in the packing. We compute the essentially
optimal solution of this type by solving a generalization of one-dimensional
knapsack for each group. Our key structural insight for medium polygons
is that such a solution is $O(1)$-approximate. To this end, we prove
that in $\OPT$ the medium polygons of each group occupy an area that
is by at most a constant factor bigger than the corresponding containers,
and that a constant fraction of these polygons fit into the containers.
In particular, we show that medium polygons with very wide bounding
boxes lie in a very small hexagonical area close to the diagonal of
the knapsack. Our routines for easy and medium polygons run in polynomial
time.

It remains to pack the \emph{hard} polygons whose bounding boxes just
fit into the knapsack or do not fit at all, even under rotation. Note
that this does not imply that the polygon itself does not fit. Our
key insight is that there can be only $O(\log N)$ such polygons in
the optimal solution, at most $O(1)$ from each group. Therefore,
we can guess these polygons in quasi-polynomial time, assuming that
$N$ is quasi-polynomially bounded. However, in contrast to other
packing problems, it is not trivial to check whether a set of given
polygons fits into the knapsack since we can rotate them by arbitrary
angles and we cannot enumerate all possibilities for the angles. However,
we show that by losing a constant factor in the approximation guarantee
we can assume that
the placement of each hard polygon comes from a precomputable polynomial size set
and hence we can guess the placements of the $O(\log N)$ hard polygons
in quasi-polynomial time. 
\begin{thm}
\label{thm:approx-qpoly}There is a $O(1)$-approximation algorithm
for 2DKP with a running time of $(nN)^{(\log nN)^{O(1)}}$.
\end{thm}

\begin{figure}
\centering{}%
\begin{tabular}{cccc}
\includegraphics[height=4.7cm]{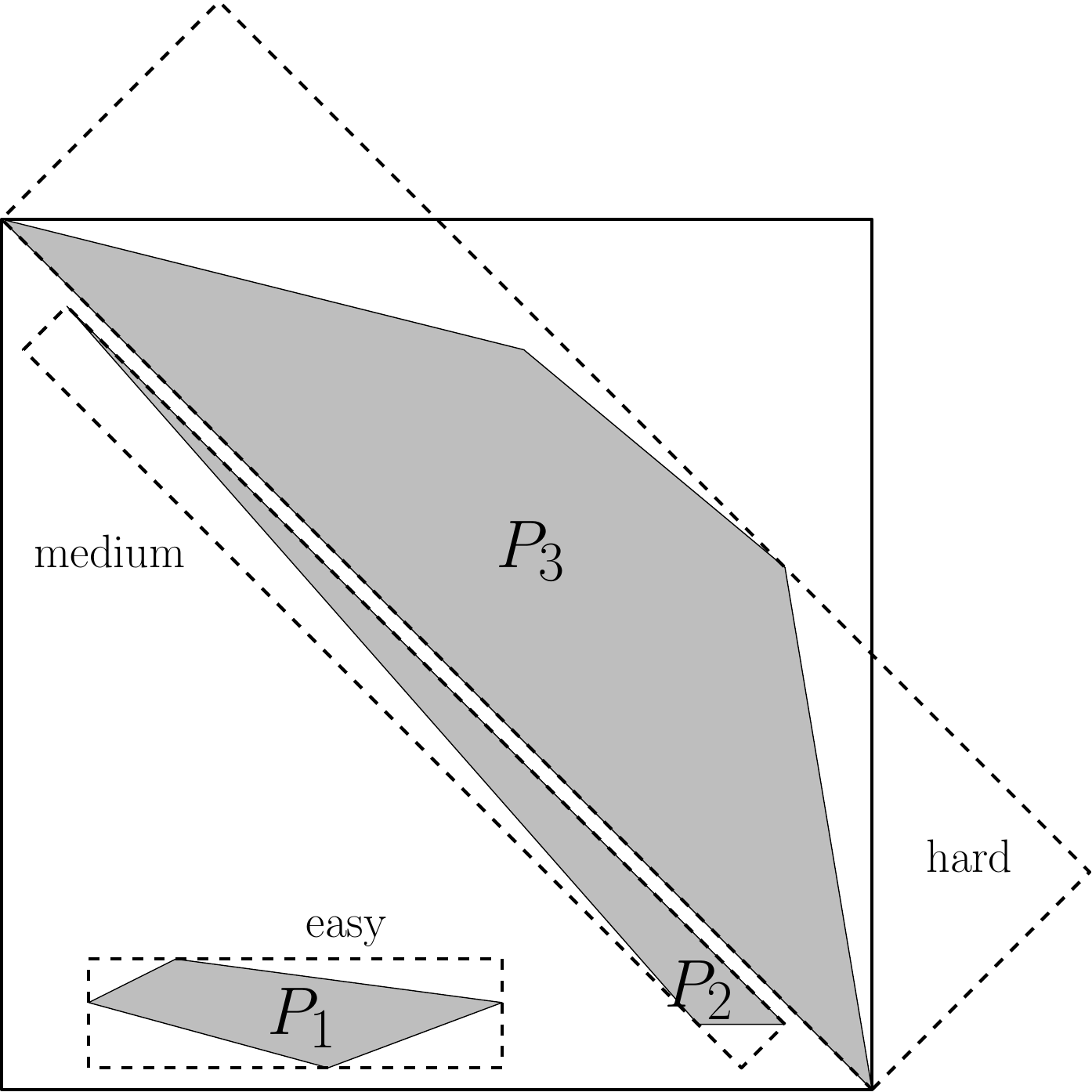} & \includegraphics[height=3.8cm]{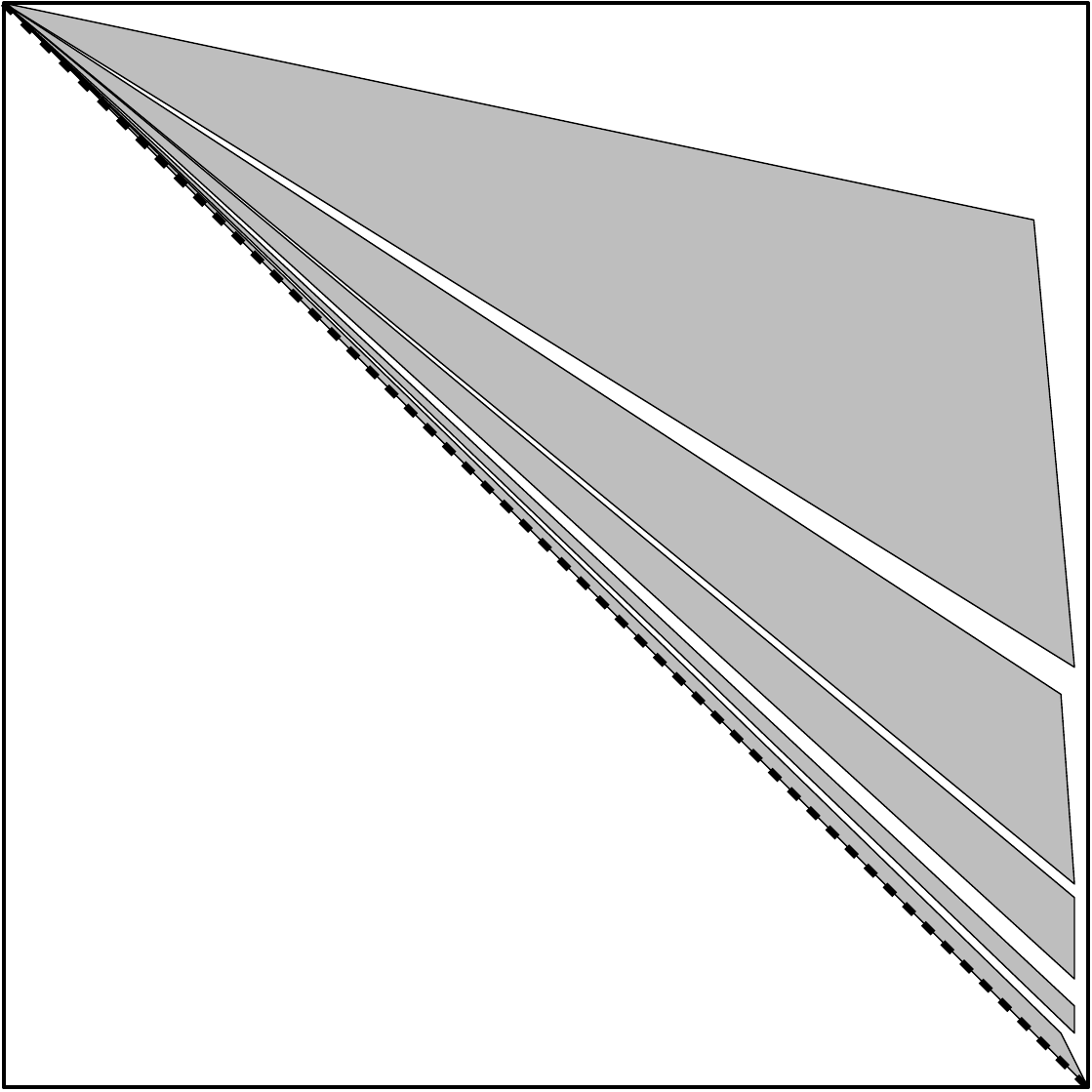} & \includegraphics[height=3.8cm]{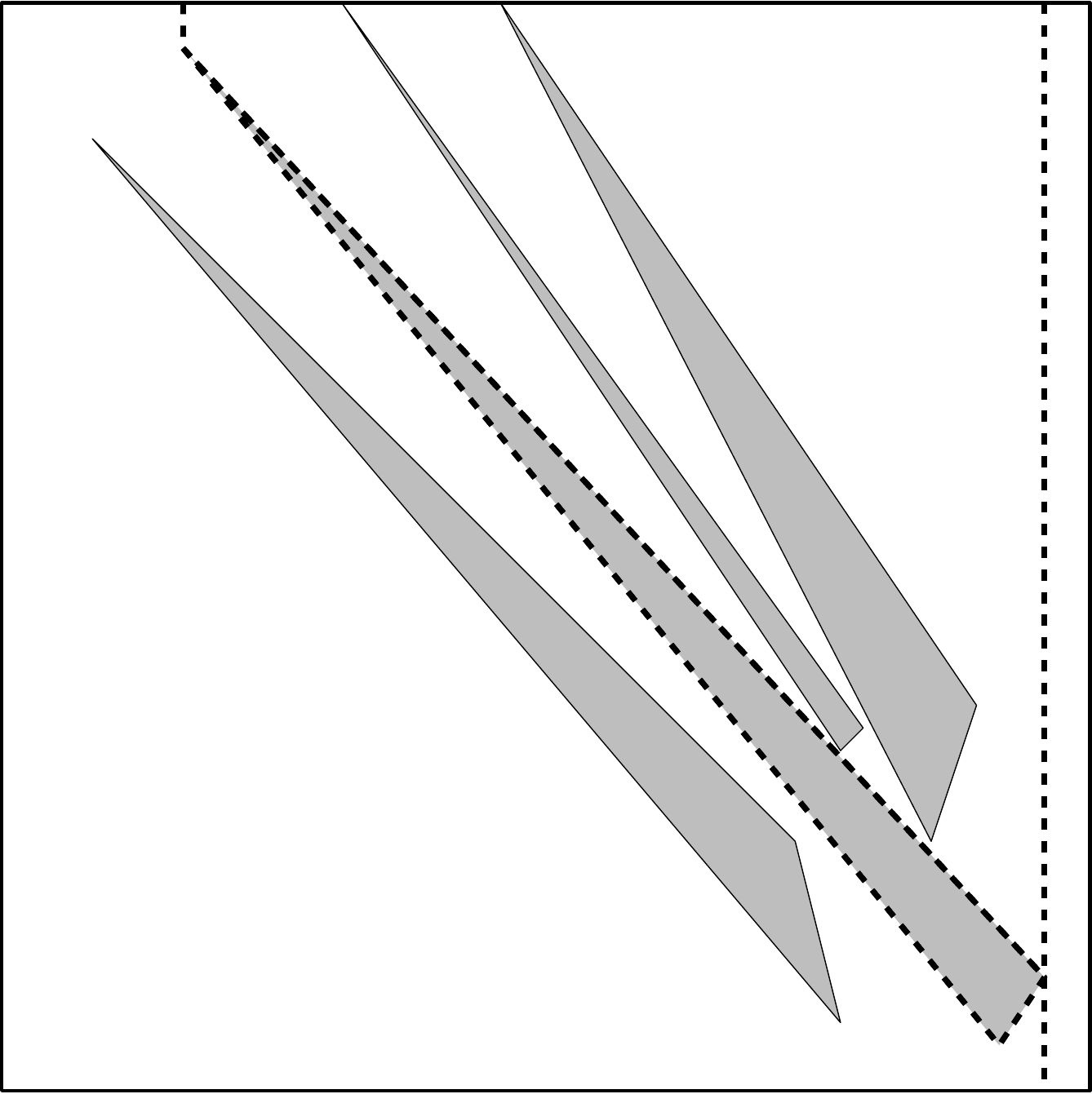}\tabularnewline
(a) & (b) & (c) \tabularnewline
\end{tabular}\caption{\label{fig:packings}(a) An easy, a medium, and a hard polygon and their bounding boxes (b): Triangles packed in a top-left packing
(c) The geometric DP subdivides the knapsack along the dashed lines
and then recurses within each resulting area.}
\end{figure}

\begin{figure}
\centering \includegraphics[height=4.5cm]{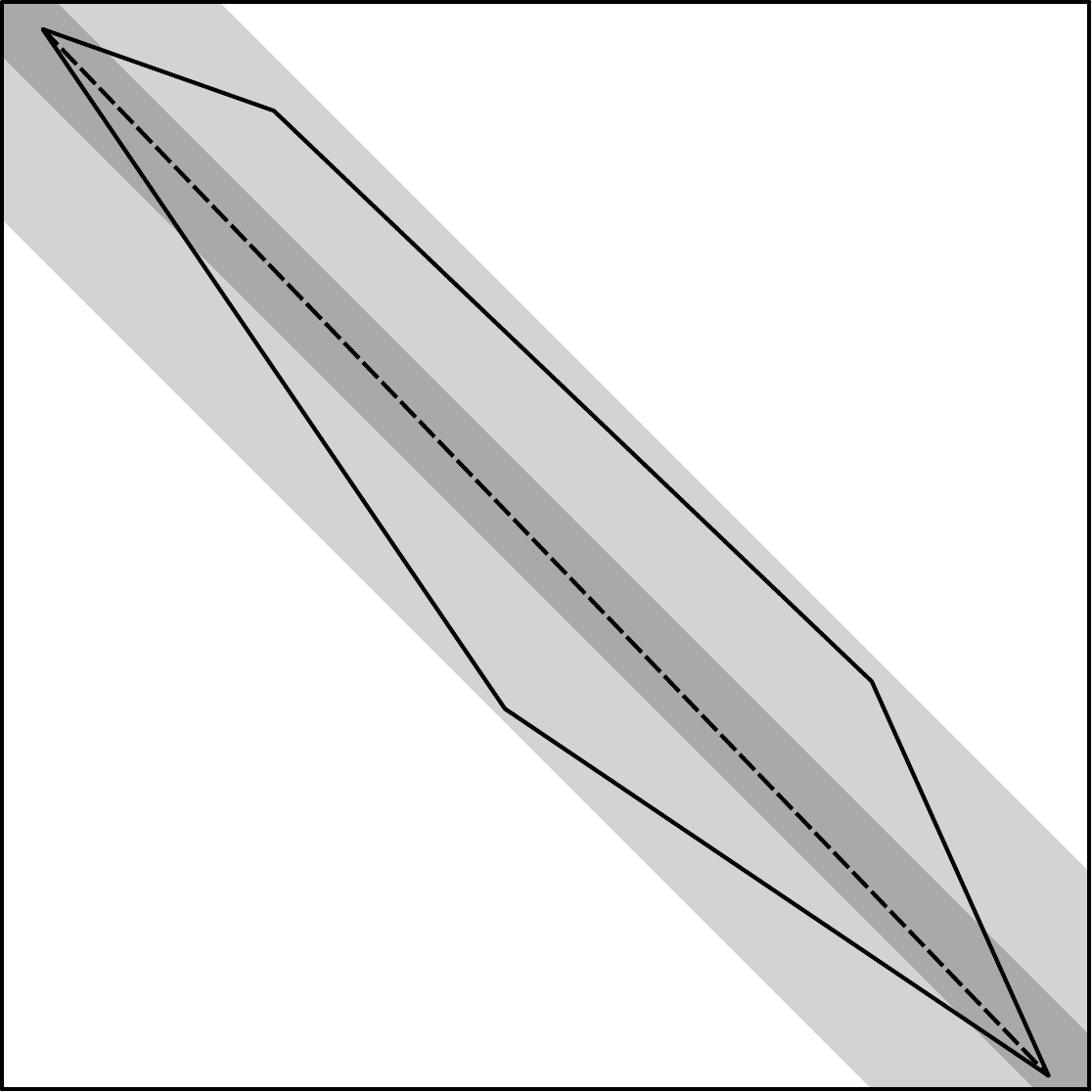}~~~~~\includegraphics[height=4.5cm]{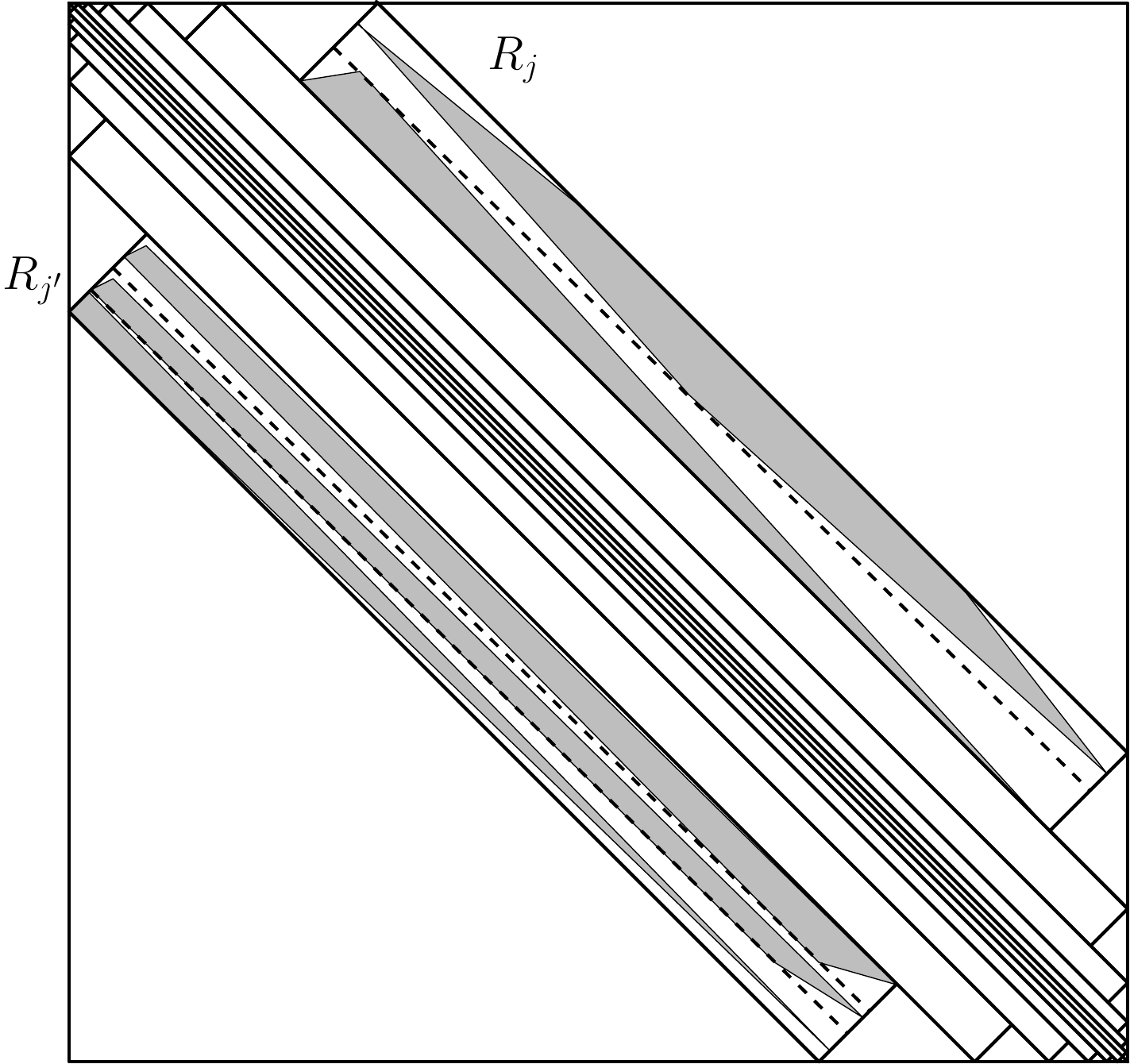}
\caption{Left: Assume that the polygon (black line segments) is a medium polygon
contained in the set $\protect\P_{j}$. Then the diagonal (dashed)
line segment must lie into the dark-gray area and the whole polygon
must be contained in the light-gray area. Right: The containers
for the medium polygons of the different groups. Within each container, the polygons are stacked
on top of each other such that their respective bounding boxes do
not intersect.}
\label{fig:hexagons-rectangles} 
\end{figure}

If all hard polygons are triangles we present even a \emph{polynomial
}time $O(1)$-approximation algorithm. We split the triangles in $\OPT$
into two types, for one type we show that a constant fraction of it
can be packed in what we call \emph{top-left-packings}, see Figure~\ref{fig:packings}b).
In these packings, the triangles are sorted by the lengths of their
longest edges and placed on top of each other in a triangular area. We devise a dynamic program (DP)
that essentially computes the most profitable top-left-packing. For
proving that this yields a $O(1)$-approximation, we need some careful
arguments for rearranging a subset of the triangles with large weight
to obtain a packing that our DP can compute. We observe that essentially all hard
polygons in $\OPT$ must intersect the horizontal line that contains
the mid-point of the knapsack. Our key insight is that if we pack
a triangle in a top-left-packing then it intersects this line to a
similar extent as in $\OPT$. Then we derive a sufficient condition
when a set of triangles fits in a top-left-packing, based on by how
much they overlap this line.

For the other type of triangles we use a geometric dynamic program.
In this DP we recursively subdivide the knapsack into subareas in
which we search for the optimal solution recursively, see Figure~\ref{fig:packings}c).
In the process we guess the placements of some triangles from $\OPT$.
Again, by losing a constant factor we can assume that for each triangle
in $\OPT$ there are only a polynomial number of possible placements.
By exploiting structural properties of this type of triangles we ensure that the
number of needed DP-cells is bounded by a polynomial. A key difficulty
is that we sometimes split the knapsack into two parts on which we
recurse independently. Then we need to ensure that we do not select
some (possibly high weight) triangle in both parts. To this end, we
globally select at most one triangle from each of the $O(\log N)$ groups (losing a constant
factor) and when we recurse, we guess for each subproblem from which
of the $O(\log N)$ groups it contains a triangle in $\OPT$. This
yields only $2^{O(\log N)}=N^{O(1)}$ guesses. 

\begin{thm}
\label{thm:approx-poly-triangles} There is a $O(1)$-approximation
algorithm for 2DKP with a running time of $(nN)^{O(1)}$ if all input
polygons are triangles.
\end{thm}

Then we study the setting of resource augmentation, i.e., we compute
a solution that fits into a larger knapsack of size $(1+\delta)N\times(1+\delta)N$ for some constant $\delta>0$
and compare ourselves with a solution that fits into the original
knapsack of size $N\times N$. We show that then the optimal solution
can contain only \emph{constantly} many hard polygons and hence we
can guess them in \emph{polynomial} time. 

\begin{thm}\label{thm:RA-ptime}
There is a $O(1)$-approximation algorithm for 2DKP under $(1+\delta)$-resource
augmentation with a running time of $n^{O_{\delta}(1)}$. 
\end{thm}
Finally, we present a quasi-polynomial time algorithm that computes
a solution of weight at least $w(\OPT)$ (i.e., we do not lose any factor
in the approximation guarantee) that is feasible under resource augmentation.
This algorithm does not use the above classification of polygons into
easy, medium, and hard polygons. Instead, we prove that if we can
increase the size of the knapsack slightly we can ensure that for
the input polygons there are only $(\log n)^{O_{\delta}(1)}$ different
shapes by enlarging the polygons suitably. Also, we show that we need
to allow only a polynomial number of possible placements and rotations
for each input polygon, \emph{without }sacrificing any polygons from
$\OPT$. Then we use a technique from~\cite{adamaszek2014qptas}
implying that there is a balanced separator for the polygons in $\OPT$
with only $(\log n)^{O_{\delta}(1)}$ edges and which intersects polygons
from $\OPT$ with only very small area. We guess the separator, guess
how many polygons of each type are placed inside and outside the separator,
and then recurse on each of these parts. Some polygons are intersected
by the balanced separator. However, we ensure that they have very
small area in total and hence we can place them into the additional
space of the knapsack that we gain due to the resource augmentation.
This generalizes a result in~\cite{adamaszek2015knapsack} for axis-parallel
rectangles. 
\begin{thm}
There is an algorithm for 2DKP under $(1+\delta)$-resource augmentation
with a running time of $n^{O_{\delta}(\log n)^{O(1)}}$ that computes
a solution of weight at least $w(\OPT)$. 
\end{thm}
In our approximation algorithms, we focus on a clean exposition of our methodology for obtaining $O(1)$-approximations, rather than
on optimizing the actual approximation ratio.

\subsection{Other related work}

Prior to the results mentioned above, polynomial time $(2+\epsilon)$-approximation
algorithms for 2DKP for axis-parallel rectangles were presented by
Jansen and Zhang~\cite{Jansen2004,jansen2004maximizing}. For the
same setting, a PTAS is known under resource augmentation in one dimension~\cite{jansen2007new}
and a polynomial time algorithm computing a solution with optimum
weight under resource augmentation in both dimensions~\cite{HeydrichWiese2019}.
Also, there is a PTAS if the weight of each rectangle equals its area~\cite{bansal2009structural}.
For squares, Jansen and Solis-Oba presented a PTAS~\cite{jansen2008ipco}.

\section{Constant factor approximation algorithms}

\label{sec:Constant-factor-approximation}

In this section we present our quasi-polynomial time $\fpoly$-approximation
algorithm for general convex polygons and our polynomial time $O(1)$-approximation
algorithm for triangles., assuming polynomially bounded input data. 
Our strategy
is to partition the input polygons $\P$ into three classes, easy,
medium, and hard polygons, and then to devise algorithms for each
class separately.

Let $K:=[0,N]\times[0,N]$ denote the given knapsack. We assume that
each input polygon is described by the coordinates of its vertices
which we assume to be integral. First, we rotate each polygon in $\P$
such that its longest diagonal (i.e., the line segment that connects
the two vertices of largest distance) is horizontal. For each polygon
$P_{i}\in\P$ denote by $(x_{i,1},y{}_{i,1}),...,(x{}_{i,k_{i}},y{}_{i,k_{i}})$
the new coordinates of its vertices. Observe that due to the rotation,
the resulting coordinates might not be integral, and possibly not
even rational. We will take this into account when we define our algorithms.
For each $P_{i}\in\P$ we define its \emph{bounding box} $B_{i}$
to be the smallest axis-parallel rectangle that contains $P_{i}$.
Formally, we define $B_{i}:=[\min_{\ell}x{}_{i,\ell},\max_{\ell}x{}_{i,\ell}]\times[\min_{\ell}y{}_{i,\ell},\max_{\ell}y{}_{i,\ell}]$.
For each polygon $P_{i}$ let $\ell_{i}:=\max_{\ell}x{}_{i,\ell}-\min_{\ell}x{}_{i,\ell}$
and $h_{i}:=\max_{\ell}y{}_{i,\ell}-\min_{\ell}y{}_{i,\ell}$. 
If necessary we will work with suitable estimates of these values later,
considering that they might be irrational and hence we cannot compute them exactly.

We first distinguish the input polygons into \emph{easy}, \emph{medium},
and \emph{hard }polygons. We say that a polygon $P_{i}$ is \emph{easy
}if $B_{i}$ fits into $K$ without rotation, i.e., such that $\ell_{i}\le N$
and $h_{i}\le N$. Denote by $\Pe\subseteq\P$ the set of easy polygons.
Note that the bounding box of a polygon in $\P\setminus\Pe$ might
still fit into $K$ if we rotate it suitably. Intuitively, we define
the medium polygons to be the polygons $P_{i}$ whose bounding box
$B_{i}$ fits into $K$ with quite some slack if we rotate $B_{i}$
properly and the hard polygons are the remaining polygons (in particular
those polygons whose bounding box does not fit at all into $K$).

Formally, for each polygon $P_{i}\in\P$ we define $h'_{i}:=\sqrt{2}N-\ell_{i}$.
The intuition for $h'_{i}$ is that a rectangle of width $\ell_{i}$
and height $h'_{i}$ is the highest rectangle of width $\ell_{i}$
that still fits into $K$.
\begin{lem}\label{lem:rotation}
Let $P_{i}\in\P$. A rectangle of width $\ell_{i}$ and height $h'_{i}$
fits into $K$ (if we rotate it by 45\degree) but a rectangle of
width $\ell_{i}$ and of height larger than $h'_{i}$ does not fit
into $K$.
\end{lem}
\begin{proof}
	We begin by proving that the rectangle $R_i=[0,\ell_i]\times [0, h_i']$ fits into $K$ when rotating by 45$^o$. To this end, just consider the placement of $R_i$ by its new vertices $v_1=(\frac{h'_i}{\sqrt 2},0)$, $v_2=(N,\frac{\ell_i}{\sqrt 2})$, $v_3=(N-\frac{h_i'}{\sqrt 2}, N)$, $v_4=(0,N-\frac{\ell_i}{\sqrt 2})$.

	We now prove the second part of the Lemma. Choose $\delta \geq 0$ maximal such that $R_i^ \delta := [0,\ell_i]\times [0,h' _i+\delta]$ fits into $K$. We aim to prove that $\delta=0$. By maximality of $\delta $ we can assume that in the placement into $K$ some vertex of $R_i^ \delta$ lies on a side of $K$. Without loss of generality we assume that $v_1$ lies on $[0,N]\times \{0\}$. Therefore $v_1=(t,0)$ for some $t \in [0,N]$. Draw the two lines that start in $v_1$ and have $45^ o$ difference with the side $[0,N]\times \{0\}$. Note that these lines intersect $K$ at $p_1:=(0,t)$ and $p_2:=(N,N-t)$, additionally $\|v_1-v_2\|=\sqrt 2t $ and $\|v_1-v_3\|=\sqrt 2 (N-t)$. Since $R_i^ \delta$ fits, these are also upper bounds on $\ell_i$ and $h_i'+\delta$ respectively. We conclude that $\ell_i \leq \sqrt 2 t$ and therefore:
	\[
	\sqrt 2 N - \ell_i +\delta = h_i' +\delta  \leq \sqrt 2 (N-t) \leq \sqrt 2 N - \ell_i,
	\]
	concluding that $\delta = 0$ and the proof of the Lemma.
\end{proof}

Hence, if $h_{i}$ is much smaller than $h'_{i}$ then $B_{i}$ fits
into $K$ with quite some slack. Therefore, we define that a polygon
$P_{i}\in\P\setminus\Pe$ is \emph{medium }if $h_{i}\le h'_{i}/8$
and \emph{hard }otherwise. Denote
by $\Pm\subseteq\P$ and $\Ph\subseteq\P$ the medium and hard
polygons, respectively. We will present $O(1)$-approximation algorithms
for each of the sets $\Pe,\Pm,\Ph$ separately. The best of the computed
sets will then yield a $O(1)$-approximation overall.

For the easy polygons, we construct a polynomial time $O(1)$-approximation
algorithm in which we select polygons such that we can pack their
bounding boxes as non-overlapping rectangles using Steinberg's
algorithm~\cite{coffman1980performance}, see Section~\ref{subsec:easy}. The approximation ratio
follows from area arguments.
\begin{lem}
\label{lem:easy-polygons}There is a polynomial time algorithm that
computes a solution $\P'_{E}\subseteq\Pe$ with $w(\OPT\cap\Pe)\le O(w(\P'_{E}))$.
\end{lem}
For the medium polygons, we obtain a $O(1)$-approximation algorithm
using a different packing strategy, see Section~\ref{subsec:Medium-polygons}.

\begin{lem}
\label{lem:medium-polygons}There is an algorithm with a running time
of $n^{O(1)}$ that computes a solution $\P'_{M}\subseteq\Pm$ with
$w(\OPT\cap\Pm)\le O(w(\P'_{M}))$.
\end{lem}
The most difficult polygons are the hard polygons. First, we show
that in quasi-polynomial time we can obtain a $O(1)$-approximation
for them, see Section~\ref{subsec:Hard-polygons}.
\begin{lem}
\label{lem:hard-polygons}There is an algorithm with a running time
of $(nN)^{(\log nN)^{O(1)}}$ that computes a solution $\P'_{H}\subseteq\Ph$
with $w(\OPT\cap\Ph)\le O(w(\Ph'))$.
\end{lem}
Combining Lemmas~\ref{lem:easy-polygons}, \ref{lem:medium-polygons},
and \ref{lem:hard-polygons} yields Theorem~\ref{thm:approx-qpoly}.
If all polygons are triangles, we obtain a $O(1)$-approximation even
in polynomial time. The following lemma is proved in Section~\ref{subsec:Hard-triangles}
and together with Lemmas~\ref{lem:easy-polygons} and \ref{lem:medium-polygons}
implies Theorem~\ref{thm:approx-poly-triangles}.
\begin{lem}
\label{lem:hard-triangles}If all input polygons are triangles, then
there is an algorithm with a running time of $(nN)^{O(1)}$ that computes
a solution $\P'_{H}\subseteq\Ph$ with $w(\OPT\cap\Ph)\le O(w(\Ph'))$.
\end{lem}
Orthogonal to the characterization into easy, medium and hard polygons,
we subdivide the polygons in $\P$ further into classes according
to the respective values $\ell_{i}$. More precisely, we do this according
to their difference between $\ell_{i}$ and the diameter of $K$,
i.e., $\sqrt{2}N$. Formally, for each $j\in\Z$ we define $\P_{j}:=\{P_{i}\in\P|\ell_{i}\in[\sqrt{2}N-2^{j},\sqrt{2}N-2^{j-1})\}$
and additionally $\P_{-\infty}:=\{P_{i}\in\P|\ell_{i}=\sqrt{2}N\}$.
Note that for each polygon $P_i \in \P$ we can compute the group $\P_{j}$ even though $\ell_i$ 
might be irrational.

\subsection{\label{subsec:easy}Easy polygons}

We present a $O(1)$-approximation algorithm for the polygons in $\Pe$.
First, we show that the area of each polygon is at least half of the
area of its bounding box. We will use this later for defining lower
bounds using area arguments. 
\begin{lem}
	\label{lem:area-bound}For each $P_{i}\in\P$ it holds that $\area(P_{i})\ge\frac{1}{2}\area(B_{i})$. 
\end{lem}
\begin{proof}
	If $P_{i}$ is a triangle the result is clear. Suppose that $P_{i}$
	has more than three sides and call $D$ its longest diagonal. Split
	$P_{i}$ into two polygons $P,Q$ by $D$. Call $T_{P}$ and $T_{Q}$
	the triangles formed by $D$ and the vertices further away from $D$
	in $P$ and $Q$ respectively. By convexity we know that $T_{P}\subseteq P$
	and $T_{Q}\subseteq Q$. We conclude by noting that $\frac{1}{2}\area(B_{i})=\area(T_{P}\cup T_{Q})\leq\area(P\cup Q)=\area(P_{i})$. 
\end{proof}
On the other hand, it is known that we can pack any set of axis-parallel
rectangles into $K$, as long as their total area is at most $\area(K)/2$
and each single rectangle fits into $K$.
\begin{thm}[\cite{steinberg1997strip}]
	\label{thm:pack-bounding-boxes}Let $r_{1},...,r_{k}$ be a set of
	axis-parallel rectangles such that $\sum_{i=1}^{k}\area(r_{i})\le\area(K)/2$
	and each individual rectangle $r_{i}$ fits into $K$. Then there is a polynomial
	time algorithm that packs $r_{1},...,r_{k}$ into $K$.
\end{thm}
We first compute (essentially) the most profitable set of polygons
from $\P_{E}$ whose total area is at most $\area(K)$ via a reduction
to one-dimensional knapsack. 
\begin{lem}
	\label{lem:large-profit}In time $(\frac{n}{\epsilon})^{O(1)}$
	we can compute a set of
	polygons $\P'\subseteq\P_{E}$ such that $w(\P')\ge(1-\epsilon)w(\OPT\cap\Pe)$
	and $\sum_{P_{i}\in\Pe}\area(P_{i})\le\area(K)$.
\end{lem}
\begin{proof}
	We define an instance of one-dimensional knapsack with a set of items
	$I$ where we introduce for each polygon $P_{i}\in\Pe$ an item $a_{i}\in I$
	with size $s_{i}:=\area(P_{i})$ and profit $p_{i}:=w_{i}$ and define
	the size of the knapsack to be $\area(K)$. We apply the FPTAS in \cite{JinKnapsack}
	on this instance and obtain a set of items $I'\subseteq I$
	such that $p(I'):=\sum_{i\in I'}p_{i}\ge(1-\epsilon)\OPT(I)$ where
	$\OPT(I)$ denote the optimal solution for the set of items $I$,
	given a knapsack of size $\area(K)$. We define $\P':=\{P_{i}\in\Pe|a_{i}\in I'\}$.
\end{proof}
The idea is now to partition $\P'$ into at most 7 sets $\P'_{1},...,\P'_{7}$.
Hence, one of these sets must contain at least a profit of $w(\P')/7$.
We define this partition such that each set $\P'_{j}$ contains only
one polygon or its polygons have a total area of at most $\area(K)/4$.
\begin{lem}
	\label{lem:large-subset}Given a set $\P'\subseteq\Pe$ with $\sum\limits _{P_{i}\in\Pe}\area(P_{i})\le\area(K)$.
	In polynomial time we can compute a set $\P''\subseteq\P'$ with $w(\P'')\ge\frac{1}{7}w(\P')$
	and additionally $\sum\limits _{P_{i}\in\P''}\area(P_{i})\le\area(K)/4$ or $|\P''|=1$.
\end{lem}

\begin{proof}
	Note that every set of rectangles that fit into the Knapsack with
	total area less than $\frac{1}{2}\area(K)$ can be packed into the
	Knapsack by Theorem \ref{thm:pack-bounding-boxes}. Therefore, any set polygons of total area
	less than $\frac{1}{4}\area(K)$ can be put into their bounding boxes
	and, since the total area of these boxes has at most doubled, the
	convex polygons can be packed accordingly. Moreover if the height
	and width of the bounding box can be computed in polynomial time,
	this placement can also be computed in polynomial time.
	
	Sort $\P'=\{P'_{1},\dots,P'_{l}\}$ decreasingly by area and define $\C_{1}=\{P'_{1}\}$, $\C_{2}=\{P'_{2}\}$
	and $\C_{3}=\{P'_{3}\}$. We now partition $\P'\setminus\{C_{1}.C_{2},C_{3}\}$ obtaining parts $\C_{4},\dots,C_{p}$ as follows. We begin by defining a sequence of $\ell_j$:
	\begin{align*}
	\ell_0 & = 0, \\
	\ell_j & = \max \left\{ k\in\{\ell_{j-1}, \dots , \ell \} : \sum_{i=\ell_{j-1}+1}^k \area(P_i')\leq \frac{1}{4} \area (K)   \right\},
	\end{align*}
	and consider the partition of $\P'\setminus\{C_{1},C_{2},C_{3}\}$ into parts of the form $\{c_{l'_{j}+1},\dots,c_{l'_{j+1}}\}$.
	Suppose for the sake of contradiction that $p\geq8$. Note that $\sum\limits _{C\in\C_{i}}\area(C)\geq\frac{1}{4}\area(K)-\area(P'_{3})$
	for each $i$ such that $4\leq i<p$, therefore: 
	\[
	\area(K)\geq\sum_{i=1}^{p-1}\sum_{C\in\C_{i}}\area(C)\geq3\cdot\area(P'_{3})+\sum_{i=4}^{p-1}\sum_{C\in\C_{i}}\area(C)\geq
	\]
	\[
	3\cdot\area(P'_{3})+(p-4)\left(\frac{\area(K)}{4}-\area(P'_{3})\right)>\area(K),
	\]
	arriving at a contradiction. Pick $\P''$ the part with largest weight,
	therefore $w(\P'')\geq\frac{1}{7}w(\P')$ and, by construction, either
	$\P''$ is a single polygon that fits or has area at most $\frac{1}{4}\area(K)$
	and can be placed non-overlappingly into the Knapsack. 
\end{proof}
If $|\P''|=1$ we simply pack the single polygon in $\P''$ into the knapsack. Otherwise,
using Lemmas~\ref{lem:area-bound} and  \ref{lem:large-profit} and
Theorem~\ref{thm:pack-bounding-boxes} we know that we can pack the 
bounding boxes of the polygons in $\P''$ into $K$. Note that their heights and widths might be irrational.
Therefore, we slightly increase them such that these values become rational, before applying Theorem~\ref{thm:pack-bounding-boxes} to compute the actual packing. If as a result the total area of the bounding boxes exceeds $\area(K)/2$ we partition them into two sets where each set 
satisfies that the total area of the bounding boxes is at most $\area(K)/2$ or it contains only one polygon and we keep the more profitable of these two sets (hence losing a factor of 2 in the approximation ratio). This yields a $O(1)$-approximation algorithm for the easy polygons and thus proves Lemma~\ref{lem:easy-polygons}.

\subsection{\label{subsec:Medium-polygons}Medium polygons}

We describe a $O(1)$-approximation algorithm for the polygons in
$\Pm$. In its solution, for each $j\in\Z$ we will define two rectangular
containers $R_{j},R'_{j}$ for polygons in $\Pm\cap\P_{j}$, each
of them having width $\sqrt{2}N-2^{j-1}$ and height $2^{j-3}$, see
Figures~\ref{fig:hexagons-rectangles}. Let $\RR:=\cup_{j}\{R_{j},R'_{j}\}$.
First, we show that we can pack all containers in $\RR$ into $K$
(if we rotate them by 45\degree).
\begin{lem}
\label{lem:pack-containers}The rectangles in $\RR$ can be packed
non-overlappingly into $K$.
\end{lem}
\begin{proof}
	Let $j^{*}$ be the largest integer such that $\sqrt{2}N-2^{j^{*}-1}>0$.
	For $j\leq j^{*}$ we place the rectangle $R_{j}$ such that its vertices
	are at the following coordinates: $(\frac{2^{j-1}}{\sqrt{2}},N)$,
	$(N,\frac{2^{j-1}}{\sqrt{2}})$, $(N-\frac{2^{j-3}}{\sqrt{2}},\frac{3}{\sqrt{2}}2^{j-3})$,
	$(\frac{3}{\sqrt{2}}2^{j-3},N-\frac{2^{j-3}}{\sqrt{2}})$. Note that
	since: 
	\begin{align*}
	\left\Vert \left(\frac{2^{j-1}}{\sqrt{2}},N\right)-\left(N,\frac{2^{j-1}}{\sqrt{2}}\right)\right\Vert _{2} & =\sqrt{2}N-2^{j-1},\\
	\left\Vert \left(N-\frac{2^{j-3}}{\sqrt{2}},\frac{3}{\sqrt{2}}2^{j-3}\right)-\left(N,\frac{2^{j-1}}{\sqrt{2}}\right)\right\Vert _{2} & =2^{j-3},
	\end{align*}
	this is indeed the rectangle $R_{j}$.
	
	To prove that these are placed non-overlappingly, we define a family
	of hyperplanes $\{H_{j}\}_{j\leq j^{*}}$ such that for each $j$,
	$R_{j}$ and $R_{j-1}$ are in different half-spaces defined by $H_{j}$.
	Define these half-spaces as: 
	\begin{align*}
	H_{j}^{=} & =\left\{ (x,y)\in[0,N]^{2}:x+y=N+\frac{2^{j-2}}{\sqrt{2}}\right\} ,\\
	H_{j}^{\leq} & =\left\{ (x,y)\in[0,N]^{2}:x+y\leq N+\frac{2^{j-2}}{\sqrt{2}}\right\} ,\\
	H_{j}^{\geq} & =\left\{ (x,y)\in[0,N]^{2}:x+y\geq N+\frac{2^{j-2}}{\sqrt{2}}\right\} .
	\end{align*}
	Note that $R_{j+1}\subseteq H_{j}^{\geq}$ and $R_{j}\subseteq H_{j}^{\leq}$.
	Therefore, if $i<j$ we have that $R_{j}\subseteq H_{j}^{\leq}$ and
	$R_{i}\subseteq H_{i-1}^{\geq}\subseteq H_{j}^{\geq}$, making this
	a non-overlapping packing.
	
	Noting that the rectangles $\{R_{j}\}_{j\leq j^{*}}$ are packed in
	the half-space $H^{+}=\{(x,y)\in[0,N]^{2}:x+y\geq N\}$, we can pack
	the rectangles $\{R'_{j}\}_{j\leq j^{*}}$ symmetrically. 
\end{proof}
For each $j\in\Z$ we will compute a set of polygons $\P'_{j}\subseteq\Pm\cap\P_{j}$
of large weight. Within each container $R_{j},R'_{j}$ we will stack
the bounding boxes of the polygons in $\P'_{j}$ on top of each other
and then place the polygons in $\P'_{j}$ in their respective bounding
boxes, see Figure~\ref{fig:hexagons-rectangles}. In particular,
a set of items $\P''_{j}\subseteq\P_{j}$ fits into $R_{j}$ (or $R'_{j}$)
using this strategy if and only if $h(\P''_{j}):=\sum_{P_{i}\in\P''_{j}}h{}_{i}\le2^{j-3}$.
Observe that for a polygon $P_{i}\in\P_{j}$ with $P_{i}\in\Ph$ it
is not necessarily true that $h_{i}\le2^{j-3}$ and hence for hard
polygons this strategy is not suitable. We compute the essentially
most profitable set of items $\P'_{j}$ that fits into $R_{j}$ and
$R'_{j}$ with the above strategy. For this, we need to solve a variation
of one-dimensional knapsack with two knapsacks (instead of one) that
represent $R_{j}$ and $R'_{j}$. The value $h_i$ for a polygon $P_i$ 
might be irrational, therefore we work with a $(1+\epsilon)$-estimate
of $h_i$ instead. This costs only a factor $O(1)$ in the approximation guarantee.

\begin{lem}
\label{lem:compute-sets}Let $\epsilon>0$. For each $j\in\Z$ there
is an algorithm with a running time of $n^{O(\frac{1}{\epsilon})}$
that computes two disjoint sets $\P'_{j,1},\P'_{j,2}\subseteq\P_{j}\cap\Pm$
such that $h(\P'_{j,1})\le2^{j-3}$ and $h(\P'_{j,2})\le2^{j-3}$
and $w(\P^{*}{}_{j,1}\cup\P^{*}{}_{j,2}) \le O(w(\P'_{j,1}\cup\P'_{j,2}))$
for any disjoint sets $\P^{*}{}_{j,1},\P^{*}{}_{j,2}\subseteq\P_{j}\cap\Pm$
such that $h(\P{}_{j,1}^{*})\le2^{j-3}$ and $h(\P{}_{j,2}^{*})\le2^{j-3}$.
\end{lem}
\begin{proof}
First, for each polygon $P_i$ we compute an estimate for $h_i$ that
overestimates the true value for $h_i$ by at most a factor $1+\epsilon$. 
Working with this estimate instead of the real value for $h_i$ loses at 
most a factor of $O(1)$ in the profit.
    
For each $j$, note that we aim to solve a variation of one-dimensional knapsack
with two identical knapsacks instead of one. In our setting we have two knapsacks
$K_{1}$ and $K_{2}$, each with capacity $2^{j-3}$, the set of objects
to choose from are $\mathcal{I}:=\P_{j}\cap\P_{M}$ and each $i\in\mathcal{I}$
has size $h_{i}$ and profit $w_{i}$. This can be done with the algorithm in \cite{ChekuriKhanna2000}.
\end{proof}

For each $j\in\Z$ with $\P_{j}\cap\Pm\ne\emptyset$ we apply Lemma~\ref{lem:compute-sets}
and obtain sets $\P'_{j,1},\P'_{j,2}$. We pack $\P'_{j,1}$ into
$R_{j}$ and $\P'_{j,2}$ into $R'_{j}$, using that $h(\P'_{j,1})\le h(R_{j})$
and $h(\P'_{j,2})\le h(R'_{j})$. Then we pack all containers $R_{j},R'_{j}$
for each $j\in\Z$ into $K$, using Lemma~\ref{lem:pack-containers}.

Let $\P_{M}':=\bigcup_{j}\P'_{j,1}\cup\P'_{j,2}$ denote the selected
polygons. We want to show that $\P_{M}'$ has large weight; more precisely
we want to show that $w(\OPT\cap\Pm)\le O(w(\P_{M}'))$. First, we
show that for each $j\in\Z$ the polygons in $\P_{j}\cap\Pm\cap\OPT$
have bounded area. To this end, we show that they are contained inside
a certain (irregular) hexagon (see Figure~\ref{fig:hexagons-rectangles}) which has small area
if the polygons $P_{i}\in\P_{j}$ are wide, i.e., if $\ell_{i}$ is
close to $\sqrt{2}N$. The reason is that then $P_{i}$ must be placed
close to the diagonal of the knapsack 
and on the other hand $h_{i}$ is relatively small (since $P_{i}$
is medium), which implies that all of $P_{i}$ lies close to the diagonal
of the knapsack. For any object $O\subseteq\R^{2}$ we define
$\area(O)$ to be its area. 
\begin{lem}\label{lem:hexagon-bound}
For each $j$ it holds that $\area(\P_{j}\cap\Pm)\le O(\area(R_{j}\cup R_{j}'))$.
\end{lem}
\begin{proof}
	We show that we can choose
	a constant $C_{j}$ such that if a polygon $P_{i}\in\P_{j}\cap\Pm$
	is placed in the knapsack, then it is contained in the (irregular)
	hexagon $H_{1,j}$ that we define to be the hexagon with vertices
	$(0,0),(0,C_{j}),(N,N-C_{j}),(N,N),(N-C_{j},N),(C_{j}0)$,
	or in the (irregular) hexagon $H_{2,j}$ that we define to be the
	hexagon with vertices $(0,N-C_{j}),(0,N),(C_{j},N),(N-C_{j},0),(N,0),(N,C_{j})$,
	see Figure~\ref{fig:hexagons-rectangles}. To prove this, we first show that in the placement
	of $P_{i}$ inside $K$ the vertices of the longest diagonal of $P_{i}$
	essentially lie near opposite corners of $K$. We then use this to
	show that this longest diagonal of $P_{i}$ lies inside one of two
	hexagons that are even smaller than $H_{1,j}$ and $H_{2,j}$, respectively,
	see Figure~\ref{fig:hexagons-rectangles}. If additionally $P_{i}\in\Pm$ we conclude that $P_{i}$
	is placed completely within $H_{1,j}\cup H_{2,j}$ (while if $P_{i}\in\Ph$
	the latter is not necessarily true). 
	
	Given a value $r\ge0$ we define $T_{1}(r)$ as the union of the triangles
	with vertices $(0,r),(r,0),(0,0)$ and $(N-r,N),(N,N-r),(N,N)$ and
	$H_{1}(r):=\conv(T_{1}(r))$. Note that $H_{1,j}=H_{1}(C_{j})$.
	Similarly, we define $T_{2}(r)$ as the union of the triangles with
	vertices $(0,N-r),(0,N),(r,N)$ and $(N-r,0),(N,0),(N,r)$ and $H_{2}(r):=\conv(T_{2}(r))$,
	noting that $H_{2,j}=H_{2}(C_{j})$.
	\begin{claim}
		\label{lem:claims} Consider a polygon $P_{i}\in\P_{j}$. Let $P$
		be a placement of $P_{i}$ inside $K$ and let $D=\overline{v^{1}v^{2}}$
		denote the longest diagonal of $P$ for two vertices $v^{1},v^{2}$
		of $P$. Define $r_{i}:=N-\sqrt{\ell_{i}^{2}-N^{2}}$. If $H_{1}(r_{i})\cup H_{2}(r_{i})\ne K$
		then it holds that $D\subseteq H_{1}(r_{i})$ or $D\subseteq H_{2}(r_{i})$.
		If additionally $P_{i}\in\P_{M}$ then $P\subseteq H_{1}(r_{i}+\sqrt{2} 2^{j-3})$
		or $P\subseteq H_{2}(r_{i}+\sqrt{2}2^{j-3})$. 
	\end{claim}
	\begin{claimproof}[Proof of claim.]
		We can assume that $r_{i}<N-r_{i}$ since otherwise $H_{1}(r_{i})\cup H_{2}(r_{i})=K$.
		Call $H$ the polygon with vertices $(0,r_{i}),(r_{i},0),(N-r_{i},N),(N,N-r_{i}),(0,N-r_{i}),(r_{i},N),(N-r_{i},0),(N,r_{i})$.
		Note that since $H$ is a polyhedron, its diameter is given by the
		two vertices furthest apart. Therefore:	\begin{align*}
		\diam(H) & =\max\left\{ \sqrt{N^{2}+(N-2r_{i})^{2}},\sqrt{(N-r_{i})^{2}+(N-r_{i})^{2}}\right\} \\
		& \le\sqrt{N^{2}+(N-r_{i})^{2}}\\
		& =\sqrt{N^{2}+(\ell_{i}^{2}-N^{2})}\\
		& =\ell_{i}
		\end{align*}
		Hence, $v^{1}\in T_{1}(r_{i})\cup T_{2}(r_{i})$ or $v^{2}\in T_{1}(r_{i})\cup T_{2}(r_{i})$.
		Assume w.l.o.g.~that $v^{1}\in T_{1}(r_{i})\cup T_{2}(r_{i})$. Furthermore,
		assume that $v^{1}\in T_{1}(r_{i})$ and assume w.l.o.g. that $v^{1}$
		is contained in the triangle with vertices $(0,r_{i}),(r_{i},0),(0,0)$.
		
		Call $H'$ the hexagon with vertices $(0,0)$, $(N-r_{i},0)$, $(N,r_{i})$,
		$(N,N-r_{i})$, $(N-r_{i},N)$, $(0,N)$. Note that since $H'$ is
		a polyhedron, its diameter is given by the two vertices furthest apart.
		Therefore $\diam(H')=\sqrt{N^{2}+(N-r_{i})^{2}}=\sqrt{N^{2}+(\ell_{i}^{2}-N^{2})}=\ell_{i}$.
		Therefore, $v^{2}$ must lie inside the triangles with vertices $(N-r_{i},N),(N,N-r_{i}),(N,N)$
		or inside the triangle with vertices $(N-r_{i},0),(N,0),(N,r_{i})$.
		If $v^{2}$ lies in the latter triangle, then $\ell_{i}\le\sqrt{N^{2}+r_{i}^{2}}<\sqrt{N^{2}+(N-r_{i})^{2}}=\ell_{i}$
		which is a contradiction. Therefore, $v^{2}\in T_{1}(r_{i})$. This
		implies that $D\subseteq H_{1}(r_{i})$. Assume now additionally that
		$P_{i}\in\P_{M}$. Let $v$ be a vertex of $P$ with $v^{1}\ne v\ne v^{2}$.
		The distance between $v$ and $D$ is at most $h_{i}\leq \frac{1}{8}(\sqrt{2}N-\ell_{i})\leq 2^{j-3}$.
		Therefore, $v$ is contained in $H_{1}(r_{i}+\sqrt{2}2^{j-3})$
		and hence $P\subseteq H_{1}(r_{i}+\sqrt{2}2^{j-3})$. 
		
		The case that $v^{1}\in T_{2}(r_{i})$ can be handled similarly and
		in this case we conclude that $v^{2}\in T_{2}(r_{i})$, $D\subseteq H_{2}(r_{i})$,
		and $P\subseteq H_{2}(r_{i}+\sqrt{2}2^{j-3})$.
	\end{claimproof}
We choose $C_{j}:=N-\sqrt{(\sqrt{2}N-2^{j})^{2}-N^{2}}+\sqrt{2}2^{j-3}$
and note that for every placement $P$ of $P_{i}$ we have that $P\subseteq H_{1}(r_{i}+\sqrt{2}2^{j-3})\subseteq H_{1}(C_{j})=H_{1,j}$
or $P\subseteq H_{2}(r_{i}+\sqrt{2}2^{j-3})\subseteq H_{2}(C_{j})=H_{2,j}$
as $r_{i}\leq N-\sqrt{(\sqrt{2}N-2^{j})^{2}-N^{2}}$. Furthermore
the area of $H_{1,j}$ (and hence the area of $H_{2,j}$) is upper
bounded as follows:

Note that we can compute $\area(H(r))$ by computing twice the area of half the hexagon: i.e. the quadrilateral $(0,0)$, $(N,N)$, $(N-r,N)$ and $(r,N)$, we now divide this quadrilateral into two isosceles right triangles with hypotenuse $r$ and a rectangle with sides $2\sqrt N- 2r/\sqrt 2$ and $r/\sqrt 2$, obtaining the following:\[ \frac{1}{2}\area(H(r))=\frac{r^2}{2}+ \frac{r}{\sqrt 2}(\sqrt 2 N - 2r)=(2N - r) \frac{r}{2} \]
More generally if $r_2>r_1$ we can compute $\area(H(r_2)\setminus H(r_1))$ by computing twice the area of the quadrilateral $(0,r_1)$, $(r_1,N)$, $(r_2,N)$, $(0,r_2)$, and doing the same splitting into triangles and rectangles we get:
\[
\frac{1}{2}\area(H(r_2)\setminus H(r_1))=\frac{r_2-r_1}{2}(2N-r_2-r_1)
\]
We can now compute the area of $H_{1,j}$
\begin{align*}
\area(H_{1,j}) & =\area(H(C_{j}-\sqrt{2}2^{j-3}))+\area(H(C_{j})\setminus H(C_{j}-\sqrt{2}2^{j-3}))\\
& \leq \left(N+\sqrt{(\sqrt{2}N-2^{j})^{2}-N^{2}}\right)\left(N-\sqrt{(\sqrt{2}N-2^{j})^{2}-N^{2}}\right)+\sqrt{2}2^{j-3}\sqrt{(\sqrt{2}N-2^{j})^{2}-N^{2}}\\
& =(2N^{2}-(\sqrt{2}N-2^{j})^{2})+\sqrt{2}2^{j-3}(\sqrt{2}N-2^{j-1})\\
& = 2^{j+1}(\sqrt 2 N - 2^{j-1}) +\sqrt{2}2^{j-3}(\sqrt{2}N-2^{j-1})\\
& \leq \left(1+\frac{1}{8\sqrt 2}\right)2^{4}2^{j-3}(\sqrt{2}N-2^{j-1})\\
& =\left(1+\frac{1}{8\sqrt 2}\right)2^{4}\area(R_{j}\cup R_{j}')\qedhere
\end{align*}
\end{proof}
Using this, we can partition $\P_{j}\cap\Pm\cap\OPT$ into at most
$O(1)$ subsets such that for each subset $\P'$ it holds that $h(\P')\le2^{j-3}$
and hence $\P'$ fits into $R_{j}$ (and $R'_{j}$) using our packing
strategy above. Here we use that each medium polygon $P_{i}\in\P_{j}$
satisfies that $h_{i}\le2^{j-3}$.
\begin{lem}
\label{lem:profitable-sets}
For each $j\in\Z$ there are disjoint set $\P^{*}{}_{j,1},\P^{*}{}_{j,2}\subseteq\P_{j}\cap\Pm\cap\OPT$
with $w(\P_{j}\cap\Pm\cap\OPT)\le O(w(\P^{*}{}_{j,1}\cup\P^{*}{}_{j,2}))$
such that $h(\P{}_{j,1}^{*})\le2^{j-3}$ and $h(\P{}_{j,2}^{*})\le2^{j-3}$.
\end{lem}
\begin{proof}
	We begin by partitioning $\P_{j}\cap\Pm\cap\OPT=\{P'_{1},\dots,P'_{m}\}$
	into groups such that each group fits into $R_{j}\cup R_{j}'$. To
	do so, define a sequence $s_{j}$ such that: 
	\begin{align*}
	s_{0} & =0,\\
	s_{\ell} & =\max\left\{ k\in\{s_{\ell-1}+1,\dots,m\}:\sum_{i=s_{\ell-1}+1}^{k}h(P'_{i})\leq2^{\ell-3}\right\} .
	\end{align*}
	Consider now the partition of $\P_{j}\cap\Pm\cap\OPT$ into $p$ parts
	of the form $\mathcal{C}_{\ell}=\{P'_{s_{\ell}}+1,\dots,P'_{s_{\ell+1}}\}$.
	Recall that $\sqrt{2}N-2^{j}\leq\ell_{i}$. Since $h_{i}\leq  \frac{h'_{i}}{8}\leq \frac{1}{8}2^{j-3}$,
	for $1\leq\ell<p$ we have $\sum_{P_{i}\in\C_{j}}h_{i}\geq(1-\frac{1}{8})2^{j-3}$.
	Therefore:
	
	\[
	\sum_{P_{i}\in\C_{j}}\area(P_{i})\ge\frac{1}{2}\sum_{P_{i}\in\C_{j}}\area(B_{i})=\frac{1}{2}\sum_{P_{i}\in\C_{j}}h_{i}\ell_{i}\geq\frac{1}{2}(1-\frac{1}{8})2^{j-3}(\sqrt{2}N-2^{j}).
	\]
	More so, as the polygons are not in $\Pe$, we get that $\sqrt{2}N-2^{j}>N\geq N(1-\frac{2^{j-1}}{\sqrt{2}N})=\frac{1}{\sqrt{2}}\left(\sqrt{2}N-2^{j-1}\right)$.
	Using this: 
	\[
	\sum_{P_{i}\in\C_{j}}\area(P_{i})\geq\frac{1}{2}\left(1-\frac{1}{8}\right)2^{j-3}(\sqrt{2}N-2^{j})\geq\frac{1}{2\sqrt{2}}\left(1-\frac{1}{8}\right)2^{j-3}(\sqrt{2}N-2^{j-1})=\frac{(1-\frac{1}{8})}{2\sqrt{2}}\area(R_{j}\cup R'_{j}).
	\]
	By Lemma \ref{lem:hexagon-bound} we get that:
	
	\[
	2^{5}(1+\frac{1}{8\sqrt 2})\area(R_{j}\cup R'_{j})\geq\area(\P_{j}\cap\Pm\cap OPT)=\sum_{\ell=1}^{p}\sum_{P_{i}\in\C_{\ell}}\area(P_{i})\geq p\frac{(1-\frac{1}{8})}{2\sqrt{2}}\area(R_{j}\cup R'_{j}),
	\]
	concluding that $p\leq 7\cdot2^{6} $.
	Call $\C_{1}^{*}$ and $\C_{2}^{*}$ the most profitable parts of
	the partition. Using our bound on $p$ we get:
	
	\[
	w(\P_{j}\cap\Pm\cap OPT)\leq\frac{p}{2}w(\C_{1}^{*}\cup C_{2}^{*})\leq7\cdot2^5  w(\C_{1}^{*}\cup C_{2}^{*}).\qedhere
	\]
\end{proof}
By combining Lemmas~\ref{lem:pack-containers}, \ref{lem:compute-sets}
and \ref{lem:profitable-sets} we obtain the proof of Lemma~\ref{lem:medium-polygons}.

\subsection{\label{subsec:Hard-polygons}Hard polygons}

We first show that for each class $\P_{j}$ there are at most a constant
number of polygons from $\P_{j}\cap\Ph$ in $\OPT$, and that for only $O(\log N)$ classes
$\P_{j}$ it holds that $\P_{j}\cap\Ph\ne\emptyset$. 
\begin{lem}
\label{lem:Pj-properties} For each $j\in\Z$ it holds that 
 $|\P_{j}\cap\Ph\cap\OPT|\leq O(1)$.
Also, if $\P_{j}\cap\Ph\ne\emptyset$ then $j\in\{j_{\min},...,j_{\max}\}$
with $j_{\min}:=-\left\lceil \log N\right\rceil $ and $j_{\max}:=1+\left\lceil \log((\sqrt{2}-1)N)\right\rceil $.
\end{lem}
\begin{proof}
Recall that $\P_{j}:=\{i\in\P\setminus\Pe|\ell_{i}\in[\sqrt{2}N-2^{j},\sqrt{2}N-2^{j-1})\}$.
	Hence, if $j>1+\log((\sqrt{2}-1)N)$ then $\sqrt{2}N-2^{j-1}<N$ and
	therefore $\P_{j}=\emptyset$ since all polygons $P_{i}$ with $\ell_{i}\in[\sqrt{2}N-2^{j},\sqrt{2}N-2^{j-1})$
	satisfy that $\ell_{i}<N$ and hence $P_{i}\notin\P\setminus\Pe$.
	
	Note that $\ell_{i}^{2}$ must be a positive integer. Then it is sufficient
	to prove that if $j\le\frac{3}{2}-\log N$, then $[(\sqrt{2}N-2^{j})^{2},(\sqrt{2}N-2^{j-1})^{2})\cap\Z$
	is empty. We prove the stronger claim that $[(\sqrt{2}N-2^{j})^{2},2N^{2})\cap\Z=\emptyset$.
	Note that: 
	\[
	(\sqrt{2}N)^{2}-(\sqrt{2}N-2^{j})^{2}=2^{j}(2\sqrt{2}N-2^{j})<2^{j+\frac{3}{2}}N\leq1
	\]
	From which we conclude.
	
	We now prove the second part of the lemma by breaking the polygons
	into two sets. Define $\F\subseteq\P_{j}\cap\Ph\cap\OPT$ as the polygons
	that fit completely into $H(C_{j})$ as in Claim \ref{lem:claims}
	and $\overline{\F}$ as $(\P_{j}\cap\Ph\cap\OPT)\setminus\F$. We
	begin with the polygons in $\F$. Recall that by Lemma $\ref{lem:area-bound}$
	$\area(P_{i})\geq\frac{1}{2}\ell_{i}h_{i}$ using additionally that
	$P_{i}\in\P_{j}\cap\Ph$ we get: 
	\[
	\area(P_{i})\geq\frac{1}{2}\ell_{i}h_{i}>\frac{1}{16}(\sqrt{2}N-2^{j})(\sqrt{2}N-\ell_{i})\geq\frac{1}{16}(\sqrt{2}N-2^{j})2^{j-1}.
	\]
	Therefore: 
	\[
	\area(\F)\geq\frac{1}{8}|\F|(\sqrt{2}N-2^{j})2^{j-2}
	\]
	Additionally, by Lemma \ref{lem:hexagon-bound} we have that $\area(\F)\leq2^{j+1}(\sqrt{2}N-2^{j-1})$.
	Combining these two, and recalling that by the first part of the Lemma
	we may assume that $j\leq\log\left(\sqrt{2}-1)N\right)$ we get: 
	\[
	|\F|\leq 2^6\left(\frac{\sqrt{2}N-2^{j-1}}{\sqrt{2}N-2^{j}}\right)\leq 2^6\left(\frac{\sqrt{2}N}{\sqrt{2}N-2^{j}}\right)\leq 2^6 \left(\frac{\sqrt{2}N}{\sqrt{2}N-(\sqrt{2}-1)N}\right)\leq 2^{6+\frac{1}{2}}
	\]
	We now deal with $\overline{\F}$ note that $K\setminus H(C_{j})$
	has two connected components, we show that in each one there is at
	most one polygon intersecting it. These two are exactly the triangle
	$T_{1}$ with vertices $(0,N)$, $(N-C,N)$, $(0,C)$ and $T_{2}$
	with vertices $(N,0)$, $(N,N-C)$ and $(C,0)$.
	
	Suppose for the sake of contradiction that there exist points $p_{1},p_{2}\in T_{1}$
	belonging to some polygons $P_{1}$ and $P_{2}$. Let $D^{1}=u^{1}v^{1}$
	and $D^{2}=u^{2}v^{2}$ be the diagonals of $P_{1}$ and $P_{2}$
	respectively. Let $r_{1}$ and $r_{2}$ be as in Claim \ref{lem:claims}
	and define $s=\max(r_{1},r_{2})$. By the same Claim we get $\{u^{1},v^{1},u^{2},v^{2}\}\subseteq T(s)$.
	We now consider the rays for $i\in[2]$, $R_{i}\doteq\{p_{i}+\lambda\binom{1}{-1}:\lambda\geq0\}$.
	Note that since $s\leq C_{j}$, then $R_{1}$ and $R_{2}$ don't
	intersect $T(s)$, but by convexity of $P_{1}$ and $P_{2}$ they
	must intersect $D^{1}$ and $D^{2}$. Note that $R_{1}$ intersects
	first $D^{1}$ and then $D^{2}$ (otherwise $P_{1}$ and $P_{2}$
	would be intersecting) and $R_{2}$ intersects first $D^{2}$ and
	then $D^{1}$. Therefore $R_{1}$ and $R_{2}$ intersect $D^{1}$
	and $D^{2}$ in different order, which means that $D^{1}$ and $D^{2}$
	must intersect, a contradiction. 
\end{proof}
We describe now a quasi-polynomial time algorithm for hard polygons, i.e., we want
to prove  Lemma~\ref{lem:hard-polygons}.
Lemma~\ref{lem:Pj-properties}
implies that $|\Ph\cap\OPT|\le O(\log N)$. Therefore, we
can enumerate all possibilities for $\Ph\cap\OPT$ in time $n^{O(\log N)}$.
For each for each enumerated set $\Ph'\subseteq\Ph$ we need to check 
whether it fits into $K$. We cannot try all
possibilities for placing $\Ph'$ into $K$ since we are allowed tof
rotate the polygons in $\Ph'$ by arbitrary angles. To this end,
we show that there is a subset of $\Ph\cap\OPT$ of large weight which
contains only a single polygon or it does not use the complete space
of the knapsack but leaves some empty space. We use this empty space
to move the polygons slightly and rotate them such that each of them is
placed in one out of $(nN)^{O(1)}$ different positions that we can
compute beforehand. Hence, we can guess all positions of these polygons
in time $(nN)^{O(\log N)}$. We define that a \emph{placement
}of a polygon $P_{i}\in\P$ inside $K$ is a polygon $\tilde{P}_{i}$
such that $d+\rot_{\alpha}(P_{i})=\tilde{P}_{i}\subseteq K$ where
$d\in\R^{2}$ and $\rot_{\alpha}(P_{i})$ is the polygon that we obtain
when we rotate $P_{i}$ by an angle $\alpha$ clockwise around its
first vertex.
\begin{lem}
\label{lem:placement-hard-polygons}For each polygon $P_{i}\in\Ph$
we can compute a set of $(nN)^{O(1)}$ possible placements $\L_{i}$
in time $(nN)^{O(1)}$ such that there exists a set $\Ph'\subseteq\Ph\cap\OPT$
with $w(\Ph\cap\OPT)\le O(w(\Ph'))$ which can be packed into $K$
such that each polygon $P_{i}$ is packed according to a placement
in $\L_{i}$.
\end{lem}
\begin{proof}
Let $\epsilon>0$ to be defined later. First, we observe that there
can be only $O_{\epsilon}(1)$ classes $\P_{j}$ containing polygons
whose respective values $\ell_{i}$ are not larger than $(\sqrt{2}-\epsilon)N$;
recall that $\sqrt{2}N$ is the length of the diagonal of $K$.
\begin{proposition}
For each $\epsilon>0$ there is a constant $k_{\epsilon}\in\N$ such
that each polygon $P_{i}\in\bigcup_{j=j_{\min}}^{j_{\max}-k_{\epsilon}}\P_{j}$
satisfies that $\ell_{i}\ge(\sqrt{2}-\epsilon)N$.
\end{proposition}

Due to Lemma~\ref{lem:Pj-properties} there can be only $O_{\epsilon}(1)$
hard polygons in $\OPT\cap\bigcup_{j=j_{\max}-k_{\epsilon}+1}^{j_{\max}}\P_{j}$.
Hence, it suffices to prove the claim for the hard polygons in $\OPT':=\OPT\cap\bigcup_{j=j_{\min}}^{j_{\max}-k_{\epsilon}}\P_{j}$.
Since for each polygon $P_{i}\in\OPT'$ it holds that $\ell_{i}\ge(\sqrt{2}-\epsilon)N$
we have that in any placement of $P_{i}$ inside $K$ the line segment
defining $\ell_{i}$ (i.e., the line segment connecting the two vertices
of $P_{i}$ with maximum distance) has to have essentially a 45\degree
angle with the edges of the knapsack.

Let $L$ denote the line segment connecting $p_{L}:=\binom{0}{N/2}$,
and $p_{R}:=\binom{N}{N/2}$. Let $L_{1}$ denote the line segment
connecting $p_{M}:=\binom{N/2}{N/2}$ and $p_{R}$, and let $L_{2}$
denote the line segment connecting $p_{L}$ and $p_{M}$ (see Figure~\ref{fig:corner-face-DP-notation}).
By losing a factor 3 we can assume that each polygon in $\OPT'$ intersects
$L_{1}$ but not $L_{2}$. We group the polygons in $\OPT'$ into
three groups $\OPT^{(1)},\OPT^{(2)},\OPT^{(3)}$. We define that $\OPT^{(1)}$
contains the polygons in $\OPT'$ that have empty intersection with
$[0,\frac{1}{nN}]\times[0,N]$. We define that $\OPT^{(2)}$ contains
the polygons in $\OPT'\setminus\OPT^{(1)}$ that have empty intersection
with $[0,N]\times[0,\frac{1}{nN}]$. Finally, we define $\OPT^{(3)}:=\OPT'\setminus(\OPT^{(1)}\cup\OPT^{(2)})$.

Consider the group $\OPT^{(1)}$. If $\epsilon$ is sufficiently small
then every polygon in $\OPT'$ intersects $L$. We sort the polygons
in $\OPT'$ in the order in which they intersect $L$ from left to
right, let $\OPT'=\{Q_{1},\dots,Q_{k}\}$ denote this ordering. For
each $i\in\{1,...,k\}$ we now translate each polygon $Q_{i}$ to
the left by $\frac{n-i+1}{n^{2}N}$ units. We argue that between any
two consecutive polygons $Q_{i},Q_{i+1}$ there is some empty space
that intuitively we can use as slack. Since $Q_{i},Q_{i+1}$ are convex,
for their original placement there is a line $L'$ that separates
them. If $\epsilon$ is sufficiently small, then the line segments
defining $\ell_{i}$ and $\ell'_{i}$ have essentially a 45\degree
angle with the edges of the knapsack. Since $\ell_{i}\ge(\sqrt{2}-\epsilon)N$
and $\ell'_{i}\ge(\sqrt{2}-\epsilon)N$ this implies that also $L'$
also essentially forms a 45\degree angle with the edges of the knapsack.
After translating $Q_{i}$ and $Q_{i+1}$ we can draw not only line
separating them (like $L'$) but instead a strip separating them,
defined via two lines $L'',L'''$ whose angle is identical to the
angle of $L'$, and such that the distance between $L''$ and $L'''$
is at least $\Omega(\frac{1}{nN})$. Next, we rotate $Q_{i}$ around
one of its vertices  until the angle of the line segment defining
$\ell_{i}$ is a multiple of $\eta\frac{1}{nN^{2}}$ for some small
constant $\eta>0$ to be defined later, or one of the vertices of
$Q_{i}$ touches an edge of the knapsack. In the latter case, let
$v$ be a vertex of $Q_{i}$ that now touches an edge of the knapsack.
We rotate $Q_{i}$ around $v$ until the angle of the line segment
defining $\ell_{i}$ is a multiple of $\eta\frac{1}{nN^{2}}$ or another
vertex $v'$ of $Q_{i}$ touches an edge of the knapsack. In the latter
case we observe that two vertices of $Q_{i}$ touch an edge of the
knapsack. Since $Q_{i}$ has at most $O(N)$ vertices, there are at
most $O(N^{2})$ such orientations for $Q_{i}$. Otherwise, there
are only $nN^{2}/\eta$ possibilities for the angle of the line segment
defining $\ell_{i}$ which gives at most $O(nN^{2}/\eta)$ possible
orientations for $Q_{i}$ in total. Finally, we move $Q_{i}$ to the
closest placement with the property that the first vertex $v$ of
$Q_{i}$ is placed on a position in which both coordinates are multiples
of $\eta\frac{1}{nN}$. One can show that due to the empty space between
any two consecutive polygons $Q_{i},Q_{i+1}$ no two polygons overlap
after the movement, if $\eta$ is chosen sufficiently small. This
yields a placement for the polygons in $\OPT^{(1)}$ in which each
polygon $Q_{i}$ is placed according to one out of $(nN)^{O(1)}$
positions.

We use a symmetric argumentation for the polygons in $\OPT^{(2)}$.
Finally, we want to argue that $|\OPT^{(3)}|=O(1)$. Observe that
each polygon $P_{i}\in\OPT^{(3)}$ intersects $L_{1}$, the stripe
$[0,\frac{1}{nN}]\times[0,N]$, the stripe $[0,N]\times[0,\frac{1}{nN}]$,
but has empty intersection with $L_{2}$. Hence, the line segment
defining $\ell_{i}$ must have length at least $\sqrt{2}N-\frac{2\sqrt{2}}{nN}$.
Therefore, $P_{i}\in\P_{-\infty}$ or $P_{i}\in\P_{j}$ with $2^{j}\le\frac{2\sqrt{2}}{nN}$
and then $j\le\log(2\sqrt{2})-\log(nN)$. First assume that $P_{i}\in\P_{j}$.
If $n$ is a sufficiently large constant we have that $j<j_{\min}$
which contradicts Lemma~\ref{lem:Pj-properties}. On the other hand,
if $n=O(1)$ then the claim is trivially true since then $|\OPT|=O(1)$.
Otherwise, assume that $P_{i}\in\P_{-\infty}$. By Lemma~\ref{lem:Pj-properties}
there can be at most $O(1)$ polygons from $\P_{-\infty}$ in $\OPT$.
\end{proof}
This yields the proof of Lemma~\ref{lem:hard-polygons}.

\subsection{\label{subsec:Hard-triangles}Hard triangles}

In this section we present a $O(1)$-approximation
algorithm in \emph{polynomial }time for hard polygons assuming that they are
all triangles, i.e., we prove Lemma~\ref{lem:hard-triangles}. Slightly abusing notation, denote by $\OPT$ the set
$\Ph'$ obtained by applying Lemma~\ref{lem:placement-hard-polygons}.
We distinguish the triangles in $\OPT$ into two types: \emph{edge-facing
}triangles
 and \emph{corner-facing }triangles. Let $P_{i}\in\OPT\cap\Ph$, let 
 $e_{1},e_{2}$ denote the two longest edges of $P_{i}$, and let $v_{i}^{*}$
the vertex of $P_{i}$ adjacent to $e_{1}$ and $e_{2}$. Let $R_{i}^{(1)}$
and $R_{i}^{(2)}$ be the two rays that originate at $v_{i}^{*}$
and that contain $e_{1}$ and $e_{2}$, respectively, in the placement
of $P_{i}$ in $\OPT$. We have that $R_{i}^{(1)}\setminus\{v_{i}^{*}\}$
and $R_{i}^{(2)}\setminus\{v_{i}^{*}\}$ intersect at most one edge
of the knapsack each. If $R_{i}^{(1)}\setminus\{v_{i}^{*}\}$ and
$R_{i}^{(2)}\setminus\{v_{i}^{*}\}$ intersect the same edge of the
knapsack then we say that $P_{i}$ is \emph{edge-facing, }if one of
them intersects a horizontal edge and the other one intersects a vertical
edge we say that $P_{i}$ is \emph{corner-facing. }The next lemma
shows that there can be only $O(1)$ triangles in $\OPT\cap\Ph$ that
are neither edge- nor corner-facing, and therefore we compute a $O(1)$-approximation
with respect to the total profit of such triangles by simply selecting
the input triangle with maximum weight. 
\begin{lem}
\label{lem:all-facing} There can be at most $O(1)$ triangles in
$\OPT\cap\Ph$ that are neither edge-facing nor corner-facing. 
\end{lem}
\begin{proof}
Let $P_{i}\in\OPT\cap\Ph$ that is neither edge-facing nor corner-facing.
Assume w.l.o.g.~that both $R_{i}^{(1)}\setminus\{v_{i}^{*}\}$ and
$R_{i}^{(2)}\setminus\{v_{i}^{*}\}$ intersect a horizontal edge of
the knapsack. Let $e_{1},e_{2}$ denote the two longest edges of $P_{i}$.
Since $P_{i}$ is hard, we know that one of these edges is longer
than $N$ and therefore the other one is longer than $N/2$. Let $\alpha$
denote the angle between $e_{1}$ and $e_{2}$. It holds that $\alpha$
cannot be arbitrarily small since otherwise it cannot be that $R_{i}^{(1)}\setminus\{v_{i}^{*}\}$
and $R_{i}^{(2)}\setminus\{v_{i}^{*}\}$ intersect different horizontal
edges of the knapsack. Formally, assume w.l.o.g.~that $v_{i}^{*}$
lies in the upper half of the knapsack and that $R_{i}^{(1)}\setminus\{v_{i}^{*}\}$
intersects the bottom edge of the knapsack. Let $\beta$ denote the
angle between the horizontal line going through $v_{i}^{*}$ and $R_{i}^{(1)}$
(note that $\alpha\ge\beta$). Then one can show that $\tan\beta\ge1/2$
and thus $\alpha\ge\beta\ge\pi/8$. Therefore, $\area(P_{i})\ge\Omega(N^{2})$.
Hence, there can be at most $O(1)$ triangles in $P_{i}\in\OPT\cap\Ph$
that are neither edge-facing nor corner-facing.
\end{proof}
Let $p_{TL}$, $p_{TR}$, $p_{BL}$, and $p_{BR}$ denote the top
left, top right, bottom left, and bottom right corners of $K$, respectively,
and let $p_{M}:=\binom{N/2}{N/2}$, $p_{L}:=\binom{0}{N/2}$, and
$p_{R}:=\binom{N}{N/2}$, see Figure~\ref{fig:corner-face-DP-notation}. 
By losing a factor $O(1)$ we assume from now on that 
that $\OPT$ contains at most one hard triangle from each
group $\P_{j}$, using Lemma~\ref{lem:Pj-properties}.

Let $\OPTef\subseteq\OPT\cap\Ph$ denote the edge-facing hard triangles
in $\OPT$ and denote by $\OPTcf\subseteq\OPT\cap\Ph$ the corner-facing
hard triangles in $\OPT$. In the remainder of this section we
present now $O(1)$-approximation algorithms for edge-facing and for
corner-facing triangles in $\Ph$. By selecting the best solution
among the two we obtain the proof of Lemma~\ref{lem:hard-triangles}.

\subsubsection{Edge-facing triangles}

We define a special type of solutions called \emph{top-left-packings
}that our algorithm will compute. We will show later that there are
solutions of this type whose profit is at least a constant fraction
of the profit of $\OPTef$.

For each $t\in\N$ let $p_{t}:=p_{M}+\frac{t}{N^2}\binom{1}{0}$.
Let $\P'=\{P_{i_{1}},...,P_{i_{k}}\}$ be a set of triangles that
are ordered according to the groups $\P_{j}$, i.e., such that for
any $P_{i_{\ell}},P_{i_{\ell+1}}\in\P'$ with $P_{i_{\ell}}\in\P_{j}$
and $P_{i_{\ell+1}}\in\P_{j'}$ for some $j,j'$ it holds that $j\le j'$.
We define a placement of $\P'$ that we call a \emph{top-left-packing.
}First, we place $P_{i_{1}}$ such that $v_{i_{1}}^{*}$ concides
with $p_{TL}$ and one edge of $P_{i_{1}}$ lies on the diagonal of
$K$ that connects $p_{TL}$ and $p_{0}$. Note that there is a unique
way to place $P_{i_{1}}$ in this way. Iteratively, suppose that we
have packed triangles $\{P_{i_{1}},...,P_{i_{\ell}}\}$ such that for
each triangle $P_{i_{\ell'}}$ in this set its respective vertex $v_{i_{\ell'}}^{*}$
coincides with $p_{TL}$, see Figure~\ref{fig:packings}c). Intuitively, we pack $P_{i_{\ell+1}}$
on top of $P_{i_{\ell}}$ such that $v_{i_{\ell+1}}^{*}$ coincides
with $p_{TL}$. Let $t$ be the smallest integer such that the line
segment connecting $p_{t}$ and $p_{R}$ has empty intersection with
each triangle $P_{i_{1}},...,P_{i_{\ell}}$ according to our placement. We place $P_{i_{\ell+1}}$
such that $v_{i_{\ell+1}}^{*}$ concides with $p_{TL}$ and one of
its edges lies on the line that contains $p_{TL}$ and $p_{t}$. There
is a unique way to place $P_{i_{\ell+1}}$ in this way. We continue
until we placed all triangles in $\P'$. If all of them are placed
completely inside $K$ we say that the resulting solution
is a \emph{top-left-packing }and that $\P'$ is \emph{top-left-packable}.
We define \emph{bottom-right-packing }and \emph{bottom-right-packable
}symmetrically\emph{, }mirroring the above definition along the line
that contains $p_{BL}$ and $p_{TR}$\emph{.}

In the next lemma, we show that there is always a top-left-packable
or a bottom-right-packable solution with large profit compared to
$\Ph\cap\OPT$ or there is a single triangle with large profit. 
\begin{lem}
\label{lem:exists-good-packable-solution}There exists a solution
$\Ph^{*}\subseteq\Ph\cap\OPTef$ such that $w(\Ph\cap\OPTef)\le O(w(\Ph^{*}))$
and 
\begin{itemize}
\item $\Ph^{*}$ is top-left-packable or bottom-right-packable and for each
$j$ we have that $|\Ph^{*}\cap\P_{j}|\le1$, 
\item or it holds that $|\Ph^{*}|=1$. 
\end{itemize}
\end{lem}
A complete proof of this Lemma is the main topic of the next subsection.

We describe now a polynomial time algorithm that computes the most
profitable solution that satisfies the properties of Lemma~\ref{lem:exists-good-packable-solution}.
To find the most profitable solution $\Ph^{*}$ that satisfies that
$|\Ph^{*}|=1$ we simply take the triangle with maximum weight, let
$P_{i^{*}}$ be this triangle. We establish now a dynamic program
that computes the most profitable top-left-packable solution; computing
the most profitable bottom-right-packable solution works analogously.
Our DP has a cell corresponding to pairs $(j,t)$ with $j,t\in\Z$.
Intuitively, $(j,t)$ represents the subproblem of computing a set
$\Ph'\subseteq\Ph$ of maximum weight such that $\Ph'\cap\P_{j'}=\emptyset$
for each $j'<j$ and $|\Ph'\cap\P_{j''}|\le1$ for each $j''\ge j$
and such that $\Ph'$ is top-left-packable inside the triangular area
$T_{t}$ defined by the line that contains $p_{TL}$ and $p_{t}$,
the top edge of $K$, and the right edge of $K$.
Given a cell $(j,t)$ we want to compute a solution $DP(j,t)$ associated
with $(j,t)$. Intuitively, we guess whether the optimal solution
$\Ph'$ to $(j,t)$ contains a triangle from $\Ph\cap\P_{j}$. Therefore,
we try each triangle $P_{i}\in\Ph\cap\P_{j}$ and place it inside
$T_{t}$ such that $v_{i}^{*}$ concides with $p_{TL}$ and one of
its edges lies on the line containing $p_{TL}$ and $p_{t}$. Let
$t'(P_{i})$ denote the smallest integer such that $t'(P_{i})\ge t$
and $p_{t'(P_{i})}$ is not contained in the resulting placement of
$P_{i}$ inside $T_{t}$. We associate with $P_{i}$ the solution
$P_{i}\cup DP(j+1,t'(P_{i}))$. Finally, we define $DP(j,t)$ to be
the solution of maximum profit among the solutions $P_{i}\cup DP(j+1,t'(P_{i}))$
for each $P_{i}\in\Ph\cap\P_{j}$ and the solution $DP(j+1,t)$.

We introduce a DP-cell $DP(j,t)$ for each pair $(j,t)\in\Z^{2}$
where $j_{\min}\le j\le j_{\max}$ and $0\le t\le\log_{1+1/n}\left(\frac{N}{2}\right)$.
Note that due to Lemma~\ref{lem:Pj-properties} for all other values
of $j$ we have that $\P_{j}\cap\Ph=\emptyset$. Also note that $p_{t}\notin K$
if $t\geq N^2/2$. 
This yields at most
$(nN)^{O(1)}$ cells in total. Finally, we output the solution
$DP(j_{\min},0)$.

In the next lemma we prove that our DP computes the optimal top-left-packable
solution with the properties of Lemma~\ref{lem:exists-good-packable-solution}.
\begin{lem}
There is an algorithm with a running time of $(nN)^{O(1)}$ that
computes the optimal solution $\P'\subseteq\Ph$ such that $\P'$
is top-left-packable or bottom-right-packable and such that for each
$j$ we have that $|\P'\cap\P_{j}|\le1$. 
\end{lem}
\begin{proof}

We say that a set of triangles $S$ is a $(j,t)$-solution if it is top-left-packable inside of $T_t$, only uses items from $\bigcup_{j''\geq j}\Ph \cap \P_{j''}$ and $|\Ph\cap\P_{j''}|\le1$ for each $j''\geq j$. Let $\OPT_{j,t}$ be the $(j,t)$-solution of maximum weight. We aim to show that $\OPT_{j,t}=DP(j,t)$ for each $j_{\min} \leq j \leq j_{\max}$ and $0\leq t \leq \log_{1+\frac{1}{n}}\left(\frac{N}{2}\right)$.

We proceed by backwards induction on $j$. If $j=j_{\max}$ then $\OPT_{j,t}$ is exactly the packing of the top-left-packable triangle of maximum weight in $T_t$. Since $DP(j,t)$ tries to top-left-pack all triangles into $T_t$ it is clear that $\OPT_{j,t}=DP(j,t)$. We now deal with the case $j<j_{\max}$. By induction $DP(j,t)$ is the solution of maximum profit among $P_i \cup \OPT(j+1,t'(P_i))$ for $P_i \in \Ph \cap \P_j$ and $\OPT(j+1,t)$. Suppose, by contradiction, that there exists a $(j,t)$-solution $S$ such that $w(S)>w(DP(j,t))$. We now consider two cases. If $|S \cap \P_j| \neq \emptyset$, we select $P_{i^*} \in S \cap \P_j$ and note that:
\[w(P_{i^*}) + w(\OPT(j+1,t'(P_{i^*}))) \leq  w(DP(j,t)) < w(P_{i^*})+w(S\setminus \{ P_{i^*}\} ).\]
Therefore $w(S\setminus \{ P_{i^*}\} )> w(\OPT(j+1,t'(P_{i^*}))) $ which contradicts the optimality of $w(\OPT(j+1,t'(P_{i^*}))) $, since they are both $(j+1,t'(P_{i^*}))$-solutions. We now consider the case $|S \cap \P_j|=\emptyset$. In this case we have:
\[\OPT(j+1,t) \leq w(DP(j,t)) < w(S).\]
Hence $w(S)> w(\OPT(j+1,t))$ which contradicts the optimality of $\OPT(j+1,t)$, as they are both $(j+1,t)$-solutions. 

We conclude that $\OPT_{j,t}=DP(j,t)$ for each $j_{\min} \leq j \leq j_{\max}$ and $0\leq t \leq \log_{1+\frac{1}{n}}\left(\frac{N}{2}\right)$. In particular $DP(j_{\min},0)=\OPT_{j,0}$, as desired. Note that the bottom-right-packable case can be dealt in a similar manner, concluding the proof.
\end{proof}

We execute the above DP and its counterpart for bottom-right-packable
solutions to obtain a top-left-packable solution $\P'_{1}$ and a
bottom-right-packable solution $\P'_{2}$. We output the most profitable
solution among $\{P_{i^{*}}\},\P'_{1},\P'_{2}$. Due to Lemma~\ref{lem:exists-good-packable-solution}
this yields a solution with weight at least $\Omega(w(\Ph\cap\OPT))$. 
\begin{lem}
\label{lem:DP-EF}There is an algorithm with a running time of $(nN)^{O(1)}$
that computes a solution $\Ph'\subseteq\Ph$ such that $w(\OPTef)\le O(w(\Ph'))$. 
\end{lem}

\subsubsection{Corner-facing triangles}

We present now a $O(1)$-approximation algorithm for the corner-facing
triangles in $\OPT$, i.e., our algorithm computes a solution $\P'\subseteq\P$
of profit at least $\Omega(w(\OPTcf))$. We first establish some properties
for $\OPTcf$. We argue that by losing a constant factor we can assume
that each triangle in $\OPTcf$ intuitively faces the bottom-right
corner.
\begin{lem}\label{lem:face-bottom-right}
By losing a factor 4 we can assume that for each triangle $P_{i}\in\OPTcf$
we have that $R_{i}^{(1)}\setminus\{v_{i}^{*}\}$ intersects the bottom
edge of the knapsack and and $R_{i}^{(2)}\setminus\{v_{i}^{*}\}$
intersects the right edge of the knapsack, or vice versa.
\end{lem}
\begin{proof}
We can partition $\OPTcf$ into four groups according to which corner
the triangles in this group face. By losing a factor of 4 we keep
only the group with largest weight. Then we rotate the solution appropriately
such that the claim of the lemma holds.
\end{proof}
In the following lemma we establish a property that will be crucial
for our algorithm. For each $P_{i}\in\OPTcf$ let $R_{i}^{\mathrm{up}}$
denote the ray originating at $v_{i}^{*}$ and going upwards. We establish
that we can assume that $R_{i}^{\mathrm{up}}$ does not intersect
with any triangle $P_{i'}\in\OPTcf$, see Figure~\ref{fig:corner-face-DP-notation}. 
\begin{lem}
\label{lem:no-intersect-up}By losing a factor $O(1)$ we can assume
that for each $P_{i},P_{i'}\in\OPTcf$ it holds that $R_{i}^{\mathrm{up}}\cap P_{i'}=\emptyset$.
\end{lem}
\begin{proof}
Let $\epsilon>0$ be a constant to be defined later. By losing a factor
$O_{\epsilon}(1)$ we can assume for each triangle in $\OPTcf$ that
its longest edge has length at least $(1-\epsilon)\sqrt{2}N$. This
holds since all other triangles are contained in only $O_{\epsilon}(1)$
groups $\P_{j}$ with only $O_{\epsilon}(1)$ triangles in $\OPTcf$
in total (see Lemma~\ref{lem:Pj-properties}). Note that this implies
then for each triangle $P_{i}\in\OPTcf$ in any placement inside of
the knapsack that the vertex $v_{i}^{*}$ lies close to $p_{TL}$,
i.e., at distance of at most $O(\epsilon)N$.

Assume by contradiction that there is a triangle $P_{i'}\in\OPTcf$
with $R_{i}^{\mathrm{up}}\cap P_{i'}\ne\emptyset$. Then $v_{i'}^{*}$
must lie on the left of the vertical line $\ell$ that contains $R_{i}^{\mathrm{up}}$.
If both other vertices of $P_{i'}$ lie on the right of $\ell$ then
this contradicts that $P_{i'}$ faces the bottom right corner. Otherwise,
we can choose $\epsilon$ sufficiently small such that $v_{i}^{*}$
is contained in the placement of $P_{i'}$ inside the knapsack since
both $v_{i}^{*}$ and $v_{i'}^{*}$ are close to $p_{TL}$ and $v_{i'}^{*}$
is incident to one edge with length at least $N$ and to one edge
with length at least $N/2$. This is a contradiction
\end{proof}
\begin{figure}
\centering 
\includegraphics[height=4.5cm]{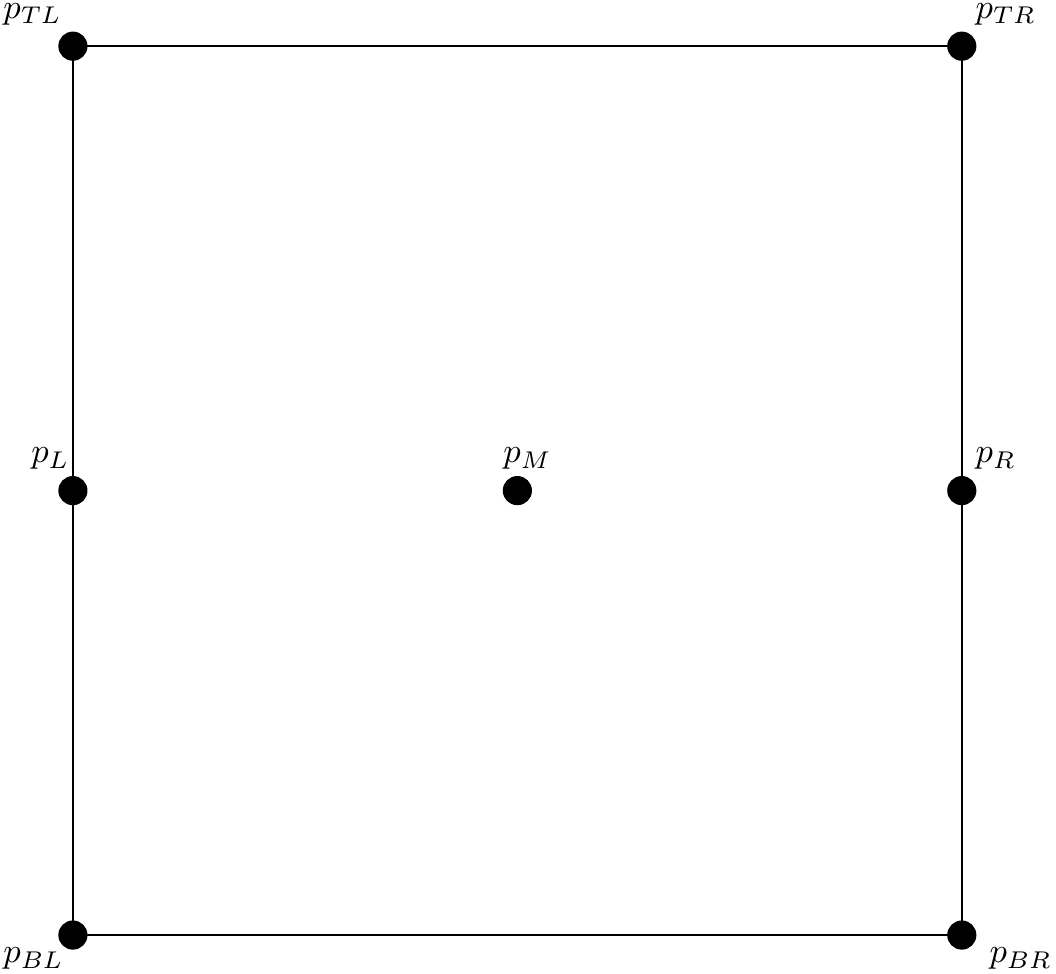}~~~~~~~~\includegraphics[height=4.5cm]{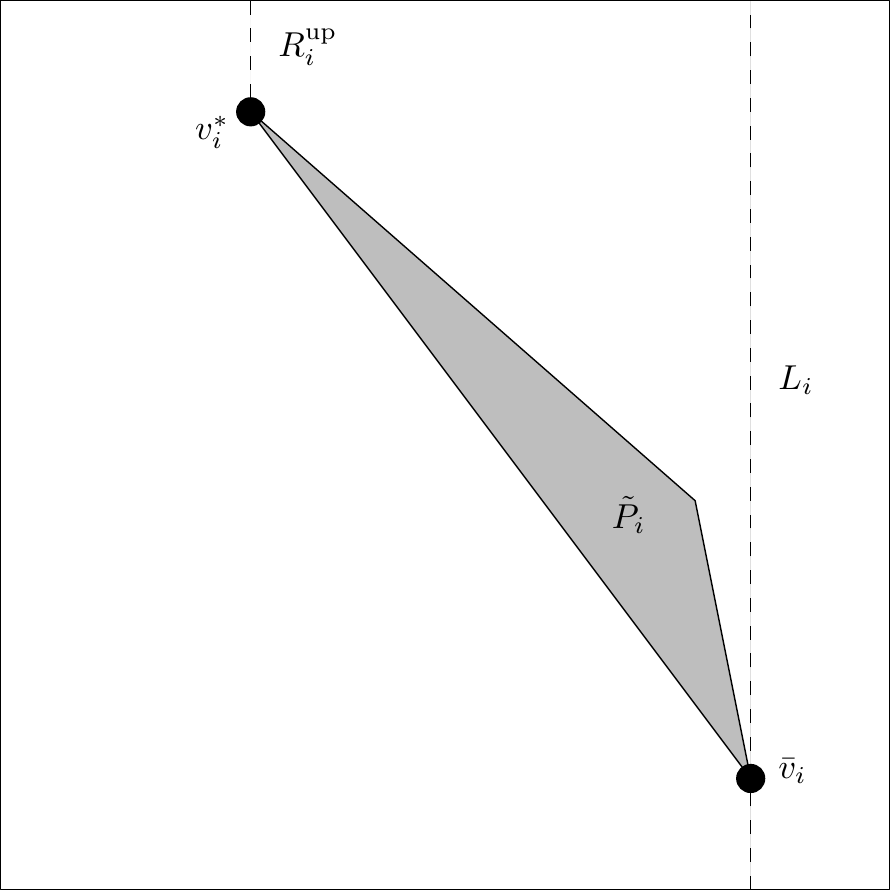}
\caption{Left: The points $p_{TL},p_{TR},p_{BL},p_{BR}, p_L, p_M, p_R$. 
Right: A corner-facing triangle, its vertices $v_i^{*}$ and $\bar{v}_i$, and the lines $L_i$ and $R_{i}^{\mathrm{up}}$.}
\label{fig:corner-face-DP-notation} 
\end{figure}

Our algorithm is a dynamic program that intuively guesses the placements
of the triangles in $\OPTcf$ step by step. To this end, each DP-cell
corresponds to a subproblem that is defined via a part $K'\subseteq K$
of the knapsack and a subset of the groups $J\subseteq\left\{ j_{\min},...,j_{\max}\right\} $.
The goal is to place triangles from $\bigcup_{j\in J}\P_{j}$ of maximum
profit into $K'$. Formally, each DP-cell is defined by up to two
triangles $P_{i},P_{i'}$, placements $\tilde{P}_{i},\tilde{P}_{i'}$
for them, and a set $J\subseteq\left\{ j_{\min},...,j_{\max}\right\} $;
if the cell is defined via exactly one triangle $P_{i}$ then there
is also a value $\mathrm{dir}\in\{\lleft,\mmid\}$. The corresponding
region $K'$ is defined as follows: if the cell is defined via zero
triangles then the region is the whole knapsack $K$. Otherwise, let
$\bar{v}_{i}$ denote the right-most vertex of $\tilde{P}_{i}$, i.e.,
the vertex of $\tilde{P}_{i}$ that is closest to the right edge of
the knapsack (see Figure~\ref{fig:corner-face-DP-notation}). 
Let $L_{i}$ denote the vertical line that goes through $\bar{v}_{i}$
(and thus intersects the top and the bottom edge of the knapsack).
If the cell is defined via one triangle $P_{i}$ then observe that
$K\setminus(\tilde{P}_{i}\cup R_{i}^{\mathrm{up}}\cup L_{i})$ has
three connected components, 
\begin{itemize}
\item one on the left, surrounded by $R_{i}^{\mathrm{up}}$, parts of $\tilde{P}_{i}$,
the left edge of the knapsack, and parts of the top and bottom edge
of the knapsack 
\item one on the right, surrounded by $L_{i}$, the right edge of the knapsack,
and parts of the top and bottom edge of the knapsack, and 
\item one in the middle, surrounded by the top edge of the knapsack, $\tilde{P}_{i}$,
and $L_{i}$. 
\end{itemize}
If $\dir=\lleft$ then the region of the cell equals the left component,
if $\dir=\mmid$ then the region of the cell equals the middle component.
Assume now that the cell is defined via two triangles $P_{i},P_{i'}$.
Assume w.l.o.g. that $\bar{v}_{i}$ is closer to the right edge of
the knapsack than $\bar{v}_{i'}$. Then $K\setminus(\tilde{P}_{i}\cup\tilde{P}_{i'}\cup R_{i}^{\mathrm{up}}\cup R_{i'}^{\mathrm{up}}\cup L_{i'})$
has one connected component that is surrounded by $\tilde{P}_{i},\tilde{P}_{i'},R_{i}^{\mathrm{up}},R_{i'}^{\mathrm{up}},L_{i'}$
and we define the region of the cell to be this component. Observe
that the total number of DP-cells is bounded by $(nN)^{O(1)}$, using
that there are only $(nN)^{O(1)}$ possible placements for each triangle.

We describe now a dynamic program that computes the optimal solution
to each cell.
\begin{figure}[t]
\begin{centering}
\includegraphics[scale=0.6]{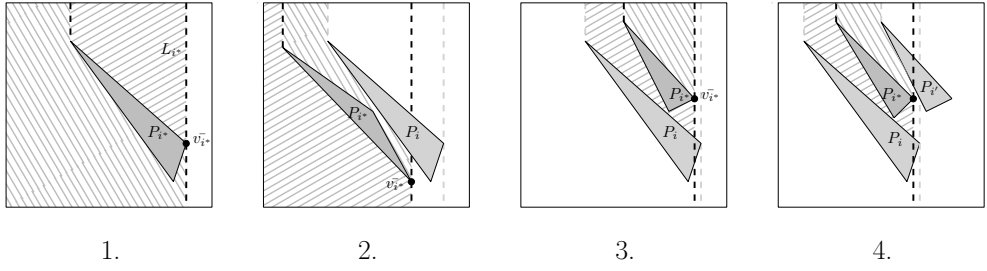}
\end{centering}
\caption{\label{fig:corner-DP-transition}The cases in the transition of the
DP for corner-facing triangles (see Lemma~\ref{lem:DP-transition}).}
\end{figure}
Assume that we are given a cell $C$ for which we want to compute
the optimal solution. We guess the triangle $P_{i^{*}}$ in the optimal
solution to this cell such that $\bar{v}_{i^{*}}$ is closest to the
right edge of the knapsack, and its placement $\tilde{P}_{i^{*}}$
in the optimal solution to $C$. Let $j^{*}$ such that $P_{i^{*}}\in\P_{j^{*}}$.
We will prove in the next lemma that the optimal solution to $C$
consists of $P_{i^{*}}$ and the optimal solutions to two other DP-cells,
see Figure~\ref{fig:corner-DP-transition}.

\begin{lem}
\label{lem:DP-transition}Let $C$ be a DP-cell, let $J\subseteq\left\{ j_{\min},...,j_{\max}\right\} $,
and let $P_{i}\in\P_{\ell},P_{i'}\in\P_{\ell'}$ be two triangles
with $\ell<\ell'$ and let $\tilde{P}_{i},\tilde{P}_{i'}$ be placements
for them. Then there are disjoint sets $J',J''\subseteq J$ such that
\begin{enumerate}
\item if $C=(J)$, then its optimal solution consists of $P_{i^{*}}$ and
the optimal solutions to the cells $(J',P_{i^{*}},\tilde{P}_{i^{*}},\lleft)$
and $(J'',P_{i^{*}},\tilde{P}_{i^{*}},\mmid)$,
\item if $C=(J,P_{i},\tilde{P}_{i},\lleft)$ then its optimal solution consists
of $P_{i^{*}}$ and the optimal solutions to the cells $(J',P_{i^{*}},\tilde{P}_{i^{*}},\lleft)$
and $(J'',P_{i},\tilde{P}_{i},P_{i^{*}},\tilde{P}_{i^{*}})$,
\item if $C=(J,P_{i},\tilde{P}_{i},\mmid)$ then its optimal solution consists
of $P_{i^{*}}$ and the optimal solutions to the cells $(J',P_{i^{*}},\tilde{P}_{i^{*}},\mmid)$
and $(J'',P_{i},\tilde{P}_{i},P_{i^{*}},\tilde{P}_{i^{*}})$,
\item if $C=(J,P_{i},\tilde{P}_{i},P_{i'},\tilde{P}_{i'})$ then the optimal
solution to $C$ consists of $P_{i^{*}}$ and the optimal solutions
to the cells $(J',P_{i},\tilde{P}_{i},P_{i^{*}},\tilde{P}_{i^{*}})$
and $(J'',P_{i^{*}},\tilde{P}_{i^{*}},P_{i'},\tilde{P}_{i'})$.
\end{enumerate}

\end{lem}
\begin{proof}
Let $\OPT_{C}$ denote the optimal solution to the cell $C$. 

\begin{claim}
\label{lem:no-intersect-right}Consider a feasible solution $S$ for the cell $C$.
Let $P_{i^{*}}$ be the triangle in $S$ whose vertex $\bar{v}_{i^{*}}$
is closest to the right edge of the knapsack. Then it holds that $L_{i^{*}}\cap P_{i}=\emptyset$
for each triangle $P_{i}\in S$. 
\end{claim}
\begin{claimproof}[Proof of Claim.]
If this was not the case then there would be a triangle $P_{i}\in S$
that intersects $L_{i^{*}}$. Hence, $P_{i}$ has one vertex that
is closer to the right edge of the knapsack than $\bar{v}_{i^{*}}$.
In particular, $\bar{v}_{i}$ is closer to the right edge of the knapsack
than $\bar{v}_{i^{*}}$ which yields a contradiction. 
\end{claimproof}
First
assume that $C=(J)$. By Lemmas~\ref{lem:no-intersect-up} and the Claim
no triangle in $\OPT_{C}$ intersects $L_{i^{*}}$ or $R_{i^{*}}^{\mathrm{up}}$
and no triangle in $\OPT_{C}$ has a vertex on the right of $L_{i^{*}}$.
Hence, each triangle in $\OPT_{C}$ is contained in the area corresponding
to the cells $(J',P_{i^{*}},\tilde{P}_{i^{*}},\lleft)$ and $(J'',P_{i^{*}},\tilde{P}_{i^{*}},\mmid)$.
We define $J'$ to be the set of indices $j\in J$ such that in $\OPT_{C}$
there is a triangle $P_{i}\in\OPT_{C}$ contained in the area corresponding
to $(J',P_{i^{*}},\tilde{P}_{i^{*}},\lleft)$ and $J''$ similarly.
The other cases can be verified similarly
\end{proof}
We guess the sets $J',J''\subseteq J$ according to Lemma~\ref{lem:DP-transition}
and store in $C$ the solution consisting of $P_{i^{*}}$, and the
solutions stored in the two cells according to the lemma. At the end,
the cell $C=(\left\{ j_{\min},...,j_{\max}\right\} )$ (whose corresponding
region equals to $K$) contains the optimal solution.  
\begin{lem}
\label{lem:DP-CF}There is an algorithm with a running time of $(nN)^{O(1)}$
that computes a solution $\P'\subseteq\P$ such that $w(\OPTcf)\le O(w(\P'))$.
\end{lem}
\begin{proof}
Since we applied Lemma~\ref{lem:few-positions} there are only $(nN)^{O(1)}$
different placements for each triangle. Also, there are only $2^{O(\log N)}=N^{O(1)}$
possibilities for the set $J$ in the description of the DP-cell.
Therefore, the number of DP-cells is bounded by $(nN)^{O(1)}$. In
order to compute the value of a DP-cell $C$, we guess the triangle
$P_{i^{*}}$ and its corresponding placement $\tilde{P}_{i^{*}}$,
and in particular we reject a guess if $\tilde{P}_{i^{*}}$ is not
contained in the region corresponding to $C$. Also, we guess $J'$
and $J''$ for which there are only $N^{O(1)}$ possibilities each
and reject guesses which do not satisfy that $J'\subseteq J$, $J''\subseteq J$,
and that $J'\cap J''=\emptyset$. Therefore, in each DP-cell we store
a solution that is feasible. We can fill the complete DP-table in
time $(nN)^{O(1)}$. Using Lemma~\ref{lem:DP-transition} one can
show that the cell $C=(\left\{ j_{\min},...,j_{\max}\right\} )$ contains
a solution $\P'$ such that with weight at least $\Omega(w(\OPTcf))$.
\end{proof}
By combining Lemmas~\ref{lem:DP-EF} and \ref{lem:DP-CF} we obtain
the proof of Lemma~\ref{lem:hard-triangles}.

\subsubsection{\label{subsec:Existence-profitable-solution}Existence of profitable
top-left- or bottom-right-packable solution}

In this subsection we prove Lemma~\ref{lem:exists-good-packable-solution}.
Let $\epsilon>0$ be a constant to be defined later. Like in the proof
of Lemma~\ref{lem:placement-hard-polygons}, we observe that there
can be only $O_{\epsilon}(1)$ classes $\P_{j}$ containing polygons
whose respective values $\ell_{i}$ are not larger than $(\sqrt{2}-\epsilon)N$;
recall that $\sqrt{2}N$ is the length of the diagonal of $K$. Furthermore,  we have that in any
placement of $P_{i}$ inside $K$ the angle $\alpha_i$ between the line segment defining $\ell_{i}$
(i.e., the line segment connecting the two vertices of $P_{i}$ with
maximum distance) and the bottom edge of the knapsack is essentially $\frac{\pi}{4}$. This implies that $\sin \alpha_i$ and $\cos \alpha_i$ are essentially $\frac{\sqrt 2}{2}$. Finally, recall that $r_i = N-\sqrt{\ell_i^2 -N^2}$ as defined in Lemma \ref{lem:hexagon-bound} this is at most $\epsilon N$.
We summarize this in the next proposition.

\begin{proposition}\label{lem:epsilon} For each $\epsilon>0$ there is a constant $k_{\epsilon}\in\N$
such that each polygon $P_{i}\in\bigcup_{j=j_{\min}}^{j_{\max}-k_{\epsilon}}\P_{j}$
satisfies that:
\begin{enumerate}
	\item $\ell_i \geq (\sqrt 2 -\epsilon)N$,
	\item $\frac{\pi}{4}-\epsilon \leq \alpha_i \leq \frac{\pi}{4}+\epsilon$,
	\item $\frac{\sqrt{2}}{2}-\epsilon \leq \sin\alpha_i \leq \frac{\sqrt{2}}{2}+\epsilon  $,
	\item $\frac{\sqrt{2}}{2}-\epsilon \leq \cos\alpha_i \leq \frac{\sqrt{2}}{2}+\epsilon  $,
	\item $r_i\leq\epsilon N$.
\end{enumerate}
\begin{proof}
	We only prove 1 and 5 as 2, 3 and 4 are direct consequence of 1. Note that by choosing $k_\epsilon=j_{\max} - \log(\epsilon N)$ we get that:
\[
	\ell_i \geq \sqrt 2 N - 2^{j_{\max} - k_\epsilon} = (\sqrt 2 -\epsilon) N. 
\]	We now prove that $r_i\leq \epsilon N$. Let $\epsilon'= \frac{\epsilon^2}{2\sqrt 2}$. By part one of this proposition we assume that $\ell_i \geq \left(2 - \epsilon'\right)N$. Therefore:
	\begin{align*}
	r_i & \leq \left(1-\sqrt{(\sqrt 2 - \epsilon ')^2 -1}\right)N \\
	& \leq (1 - \sqrt{1-2\sqrt 2 \epsilon'}) N \\
	& \leq \sqrt{2\sqrt 2 }\sqrt{\epsilon'}N = \epsilon N.\qedhere
	\end{align*}\end{proof}
\end{proposition}

Due to Lemma~\ref{lem:Pj-properties} there can be only $O_{\epsilon}(1)$
hard polygons in $\OPT\cap\bigcup_{j=j_{\max}-k_{\epsilon}+1}^{j_{\max}}\P_{j}$.
Hence, it suffices to prove the claim for the hard polygons in $\OPT_{W}:=\OPTef\cap\bigcup_{j=j_{\min}}^{j_{\max}-k_{\epsilon}}\P_{j}$
since otherwise the second case of Lemma~\ref{lem:exists-good-packable-solution}
applies if we define that $\P_{H}^{*}$ contains the polygon in $\OPTef$
of maximum weight. Note that it holds that each triangle $P_{i}\in\OPT_{W}$
intersects the line segment $L$ that we define to be the line segment
that connects $p_{L}$ with $p_{R}$. Let $L_{1}$ denote the subsegment
of $L$ that connects $p_{M}$ with $p_{R}$ and let $L_{2}$ denote
the line segment connecting $p_{L}$ with $p_{M}$. Now each triangle
in $\OPT_{W}$ either overlaps $p_{M}$ or intersects $L_{1}$ but
not $L_{2}$ or it intersects $L_{2}$ but not $L_{1}$. Therefore,
by losing a factor of 3 we can restrict ourselves to one of these
cases.
\begin{lem}
\label{lem:factor-3-L1}If $\epsilon$ is sufficiently small then
by losing a factor 3 we can assume that for each triangle $P_{i}\in\OPT_{W}$
we have that $P_{i}\cap L=P_{i}\cap L_{1}$ or that $|\OPT_{W}|=1$. 
\end{lem}

\begin{proof}
There can be at most one triangle $P_{i^{*}}\in\OPT_{W}$ that overlaps
$p_{M}$. Each other triangle $P_{i}\in\OPT_{W}$ satisfies that $P_{i}\cap L=P_{i}\cap L_{1}$
or that $P_{i}\cap L=P_{i}\cap L_{2}$. If the triangles $P_{i}\in\OPT_{W}$
satisfying $P_{i}\cap L=P_{i}\cap L_{1}$ have a total weight of at
least $\frac{1}{3}w(\OPT{}_{W})$ or if $w_{i^{*}}\ge\frac{1}{3}w(\OPT_{W})$
then we are done. Otherwise the triangles satisfying that $P_{i}\cap L=P_{i}\cap L_{2}$
have a total weight of at least $\frac{1}{3}w(\OPT{}_{W})$ and we
establish the claim of the lemma by rotating $\OPT$ by 180\degree.
\end{proof}
If $|\OPT_{W}|=1$ then we are done. Therefore, assume now that $P_{i}\cap L=P_{i}\cap L_{1}$
for each $P_{i}\in\OPT_{W}$. In the next lemma we prove that by losing
a factor of $O(1)$ we can assume that the triangles in $\OPT_{W}$
intersect $L_{1}$ in the order of their groups $\P_{j}$ (assuming
that $\epsilon$ is a sufficiently small constant). We call such a
solution group-respecting as defined below.
\begin{defn}
Let $\P'=\{P_{i_{1}},...,P_{i_{k}}\}$ be a solution in which each
triangle intersects $L_{1}$ and assume w.l.o.g.~that the triangles
in $\P'$ intersect $L_{1}$ in the order $P_{i_{1}},...,P_{i_{k}}$
when going from $p_{M}$ to $p_{R}$. We say that $\P'$ is \emph{group-respecting
}if for any two triangles $P_{i_{\ell}},P_{i_{\ell+1}}\in\P'$ with
$P_{i_{\ell}}\in\P_{j}$ and $P_{i_{\ell+1}}\in\P_{j'}$ for some
$j,j'$ it holds that $j\le j'$.
\end{defn}

For each $P_{i}\in\OPT_{W}$ let $d_{i}$ denote the length of the
intersection of $P_{i}$ and $L_{1}$ in the placement of $\OPT$.
\begin{lem}
\label{lem:group-respect}If $\epsilon$ is sufficiently small, then
by losing a factor $O(1)$ we can assume that $\OPT_{W}$ is group-respecting
and that $|\OPT_{W}\cap\P_{j}|\le1$ for each $j$. 
\end{lem}

\begin{proof}
Due to Lemma~\ref{lem:Pj-properties} we lose only a factor $O(1)$
by requiring that $|\OPT_{W}\cap\P_{j}|\le1$ for each $j$. We prove
now that by losing another factor $O(1)$ we can assume that $\OPT_{W}$
is group-respecting.

Let $P_{i}\in\P_{j}\cap\P_{H}$. Let $D$ be the longest edge of $P_{i}$
in the placement of $P_{i}$ in $\OPT_{W}$. Let $\alpha$ be the
angle between $D$ and $L_{1}$.

Like in the proof of Lemma~\ref{lem:hexagon-bound} we define $r_{i}:=N-\sqrt{\ell_{i}^{2}-N^{2}}$.
Intuitively, using $r_{i}$ we can define two hexagonical areas $H_{1}(r_{i}),H_{2}(r_{i})$
close to the diagonals of the knapsack such that the longest edge
of $P_{i}$ lies within $H_{1}(r_{i})$ or within $H_{2}(r_{i})$.
Also, let $r'_{i}$ denote the length of the line segment $L_{1}\cap H_{1}(r_{i})$.
Observe that $r'_{i}=\Theta(r_{i})$.

Let $\tilde{B}_{i}$ denote the rectangle obtained by taking the bounding
box $B_{i}$ of $P_{i}$ and moving and rotating it such that one
of its edges coincides with $P_{i}$ (in the placement of $\OPT$)
inside $K$. Note that the intersection of $\tilde{B}_{i}$ and $L_{1}$
has length at most $\frac{h_{i}}{\sin\alpha}\leq \frac{1}{1-\epsilon}h_i\leq2h_{i}$ by Proposition \ref{lem:epsilon}.
Therefore, $d_{i}\le2h_{i}$. On the other hand, if $\epsilon$ is
sufficiently small then $h_{i}\le O(d_{i})$ and thus $d_{i}=\Theta(h_{i})$.
Also, it holds that $r_{i}=\Theta(h'_{i})$ and hence $d_{i}=\Omega(r_{i})$.
Let $\gamma$ be a constant such that $d_{i}\ge\gamma\cdot r_{i}$
for each $P_{i}\in\OPT_{W}$.

Then there exists a constant $\Gamma\ge1$ such that the following
holds. For any two polygons $P_{i}\in\OPT_{W}\cap P_{j}$, $P_{i'}\in\OPT_{W}\cap P_{j'}$
such that $j+\Gamma\le j'$ we have that $r'_{i}\le d_{i'}$ since
$r'_{i}\le\Omega(r_{i})\le\Omega(r_{i'})\le\Omega(d_{i})$ and for
each constant $\gamma\ge1$ we can guarantee that $r_{i}\le\gamma r_{i'}$
if we choose $\Gamma$ sufficiently large. Let $D$ be the longest
edge of $P_{i'}$. Assume w.l.o.g.~that $D$ lies within $H_{1}(r_{i})$.
However, we have that $r_{i}\le d_{i'}$. Since $P_{i}$ and $P_{i'}$
do not intersect, they must therefore be placed in a group-respecting
manner in $\OPT_{W}$.

We split $\OPT_{W}$ into $\Gamma$ groups such that for each offset
$a\in\{0,...,\Gamma-1\}$ we define $\OPT_{W}^{(a)}:=\OPT_{W}\cap\bigcup_{k\in\Z}\P_{a+k\Gamma}$.
Therefore, for each solution $\OPT_{W}^{(a)}$ it holds that for any
two distint polygons $P_{i}\in\OPT_{W}^{(a)}\cap P_{j}$, $P_{i'}\in\OPT_{W}^{(a)}\cap P_{j'}$
for values $j,j'$ it holds that $P_{i}$ and $P_{i'}$ are placed
in a group-respecting manner in $\OPT_{W}$. Then taking the most
profitable solution among the solutions $\left\{ \OPT_{W}^{(a)}\right\} _{a\in\{0,...,\Gamma-1\}}$
loses at most another factor $\Gamma=O(1)$.
\end{proof}

For each triangle $P_{i}\in\OPT_{W}$ let $v_{i}^{*}$ be the vertex
adjacent to the two longest edges of $P_{i}$ in the placement of
$P_{i}$ in $\OPT$. Also, let $\theta_i$ denote the angle at $v_i^*$. 
We note that there exist only constantly many triangles $P_{i}\in\OPT_{W}$ with $\theta_i>\epsilon$.

\begin{lem}\label{lem:theta-bo}
	There exist at most $O_\epsilon(1)$ triangles in $P_i \in \OPT_W$ such that $\theta_i > \epsilon$. 
\end{lem}
\begin{proof}
	Let $P_{i_k} \in \OPT_W$.
	Note that the triangle $T$ with angle $\theta_i$ at $v_{i}^*$ and side lengths $\frac{\ell_i}{2}$ and $\tan\theta_i\frac{\ell_i}{2}$ is contained in $P_{i}$ (otherwise $v_{i}^*$ is not adjacent to the two longest sides). Therefore:
	\begin{align*}
	\area(P_{i}) & \geq \area(T) \\
	& =\frac{\ell_i^2}{8} \tan \theta_i \\
	& > \frac{\ell_i^2}{8} \tan\epsilon \\
	& \geq \frac{(\sqrt 2 -\epsilon)^2}{8}N^2\tan \epsilon .
	\end{align*} 
	Hence there can only be at most $O_\epsilon(1)$ such triangles in $\OPT_W$.
\end{proof}

Due to Lemma \ref{lem:theta-bo} we assume now $\theta_i \leq \epsilon$ for each $P_i \in \OPT_W$. 
Furthermore, since each triangle $P_{i}\in\OPT_{W}$
is very wide, $v_{i}^{*}$ must be close to one of the four corners
of $K$ since otherwise the longest edge of $P_{i}$ does not fit
into $K$. Thus, by losing a factor 4 we assume that $v_{i}^{*}$
is close to $p_{TL}$ for each $P_{i}\in\OPT_{W}$. 
\begin{lem}
\label{lem:close_top} By losing a factor $2$ we can assume for each $P_{i}\in\OPT_{W}$ that
$\left\Vert v_{i}^{*}-p_{TL}\right\Vert _{2}\le 2\epsilon N.$

\end{lem}

\begin{proof}
Let $u,v$ be the vertices that define $\ell_i$. By Claim \ref{lem:claims} we know that $u$ or $v$ is at distance $r_i$ of some corner $v_C$ of $K$. Without loss of generality and applying Proposition \ref{lem:epsilon} we assume that $\|v - v_C\| \leq \epsilon N$. Recall that by Proposition \ref{lem:epsilon} $\|u-v\|=\ell_i \geq (\sqrt 2 -\epsilon)N$. Hence:
\begin{align*}
\| u - v_C\| & = \|u-v -(v_C-v) \| \\
& \geq \|u-v\| - \|v_C-v\| \\
& \geq (\sqrt 2 - \epsilon) N -\epsilon N \\
& = (\sqrt 2 - 2\epsilon )N 
\end{align*}

Let $B(x,r):=\{ p |  \|x-p\|\leq r \}$ and call $v_{C'}$ the corner furthest away from $v_C$. Note that $u \in K \setminus B(v_C, (\sqrt{2}-2\epsilon)N) \subseteq B(v_{C'}, 2\epsilon N)$, from which we conclude that every endpoint of the diagonal is at distance at most $2\epsilon N$ from a corner. In particular $v_{i}^{*}$ must be a distance at most $2\epsilon N$ from some corner.

We define $N_{TL}:=\{P_{i}\in\OPT_{W}|\|v_{i}^{*}-p_{TL}\|_{2}\leq 2\epsilon N \}$
and $N_{TL}$, $N_{BL}$, $N_{BR}$ in a similar fashion. These sets
partition $\OPT_{W}$ into four sets. Note that one of these sets
must have weight at least $\frac{1}{4}w(\OPT_{W})$. If this set
is $N_{TL}$ we are done. Otherwise, we simply rotate $\OPT_{W}$
accordingly. 
\end{proof}

Due to Lemma~\ref{lem:close_top}, if $\epsilon$ is sufficiently
small we have for each triangle $P_{i}\in\OPT_{W}$ that
both $R_{i}^{(1)}\setminus\{v_{i}^{*}\}$ and $R_{i}^{(2)}\setminus\{v_{i}^{*}\}$
intersect the right edge of the knapsack or both $R_{i}^{(1)}\setminus\{v_{i}^{*}\}$
and $R_{i}^{(2)}\setminus\{v_{i}^{*}\}$ intersect the bottom edge
of the knapsack. We call triangles $P_{i}$ of the former type \emph{right-facing}
triangles and we call the triangles of the latter type \emph{bottom-facing
}triangles.
\begin{proposition}
If $\epsilon$ is sufficiently small, we have that by losing a factor
of 2 we can assume that each triangle in $\OPT_{W}$ is right-facing
or bottom-facing.
\end{proposition}

Assume that $\OPT_{W}=\left\{ P_{i_{1}},...,P_{i_{|\OPT_{W}|}}\right\} $.
We partition now $\OPT_{W}$ into $g=O(1)$ groups such that each
group is top-left-packable. Then the most profitable such group yields
a $g$-approximation. We initialize $\OPT_{W}^{(1)}:=\OPT_{W}^{(2)}:=...\OPT_{W}^{(g)}:=\emptyset$
and $k:=0$. Suppose inductively that for some $k\in\N_{0}$ we partitioned
the triangles $P_{i_{1}},...,P_{i_{k-1}}$ into $\OPT_{W}^{(1)},...,\OPT_{W}^{(g)}$
such that each of these sets is top-left-packable. We argue that there
is one value $t\in\{1,...,g\}$ such that $\OPT_{W}^{(t)}\cup\{P_{i_{k}}\}$
is also top-left-packable. To this end, observe that in the top-left-packing
of each set $\OPT_{W}^{(t)}$ each triangle $P_{i}\in\OPT_{W}^{(t)}$
blocks a certain portion of $L_{1}$ such that no other triangle in
this packing can overlap this part of $L_{1}$. 
To this end, for each triangle $P_{i}\in\left\{ P_{i_{1}},...,P_{i_{k-1}}\right\} $
let $t(i)\in\N_{0}$ be the smallest integer $t$ such that if $P_{i}\in\OPT_{W}^{(g')}$
for some $g'\in\{1,...,g\}$ then in the top-left packing of $\OPT_{W}^{(g')}$
the longest edge $e$ of $P_{i}$ lies on the line that contains $p_{TL}$
and $p_{t}$. Also, let $t'(i)$ be the smallest integer $t'$ such
that $t(i)<t'$ and $P_{i}$ does not overlap the point $p_{t'}$.
Then, after placing $P_{i}$ we cannot add another triangle in a top-left-packing
to $\OPT_{W}^{(t)}$ that touches the subsegment of $L_{1}$ that
connects $p_{t(i)}$ with $p_{t'(i)}$. Hence, intuitively $P_{i}$
blocks the the latter subsegment. We define $\hat{d}_{i}:=\left\Vert p_{t'(i)}-p_{t(i)}\right\Vert _{2}$.
Our crucial insight is now that up to a constant factor, in our top-left
packing the triangle $P_{i}$ blocks as much of $L_{1}$ as it covers
of $L_{1}$ in $\OPT_{W}$. 
\begin{lem}
\label{lem:overlap-difference}If $\epsilon$ is sufficiently small
then for each triangle $P_{i}\in\OPT_{W}$ it holds that $\hat{d}_{i}=O(d_{i})$. 
\end{lem}

\begin{proof}
We argue in a similar way as in the proof of Lemma~\ref{lem:group-respect}.
Let $D$ be the longest edge of $P_{i}$ in an arbitrary placement
of $P_{i}$ inside $K$. Let $\tilde{d}_{i}$ denote the length of
the intersection of $P_{i}$ and $L_{1}$ in this placement. Let $\alpha$
be the angle between $D$ and $L_{1}$. Due to Proposition~\ref{lem:epsilon}
we can assume that $\alpha\in[\pi/4-\frac{1}{10},\pi/4+\frac{1}{10}]$.
Hence, if $\epsilon$ is a sufficiently small then the intersection
of $\tilde{B}_{i}$ and $L_{1}$ has length at most $2h_{i}$. Therefore,
$\tilde{d}_{i}\le2h_{i}$.

On the other hand, if $\epsilon$ sufficiently small then $h_{i}\le O(\tilde{d}_{i})$.
Hence, $\tilde{d}_{i}=\Theta(h_{i})$ and also $d_{i}=\Theta(h_{i})$
and therefore $\tilde{d}_{i}=\Theta(d_{i})$. Since $h_{i}\ge h'_{i}/8$
this implies that $d_{i}=\Omega(h'_{i})$. Also, it holds that $r_{i}=\Theta(h'_{i})$
and $D$ lies in $H_{1}(r_{i})$ or $H_{2}(r_{i})$. Therefore, if
$n$ is sufficiently large (which implies that the points $p_{t}$
are sufficiently dense on $L_{1}$) we have that $\hat{d}_{i}=O(\tilde{d}_{i})=O(d_{i})$. 
\end{proof}

{} Lemma \ref{lem:overlap-difference} implies that if $g$ is a sufficiently
large constant then there is a value $t\in\{1,...,g\}$ such that
$\sum_{P_{i_{\ell}}\in\OPT_{W}^{(t)}}\hat{d}_{i_{\ell}}\le\sum_{P_{i_{\ell}}\in\left\{ P_{i_{1}},...,P_{i_{k-1}}\right\} }d_{i_{\ell}}$.
Hence, in the top-left packing for $\OPT_{W}^{(t)}$ (which at this
point contains only triangles from $\left\{ P_{i_{1}},...,P_{i_{k-1}}\right\} $)
the triangles block less of $L_{1}$ than the amount of $L_{1}$ that
the triangles $P_{i_{1}},...,P_{i_{k-1}}$ cover in $\OPT_{W}$. On
the other hand, we know that in $\OPT_{W}$ the triangle $P_{i_{k}}$
is placed such that it intersects $L_{1}$ further on the right than
any triangle in $P_{i_{1}},...,P_{i_{k-1}}$ due to Lemma~\ref{lem:group-respect}.
Using this, in the next lemmas we show that we can add $P_{i_{k}}$
to $\OPT_{W}^{(t)}$.
\begin{figure}[t]
	\centering \includegraphics[height=7cm]{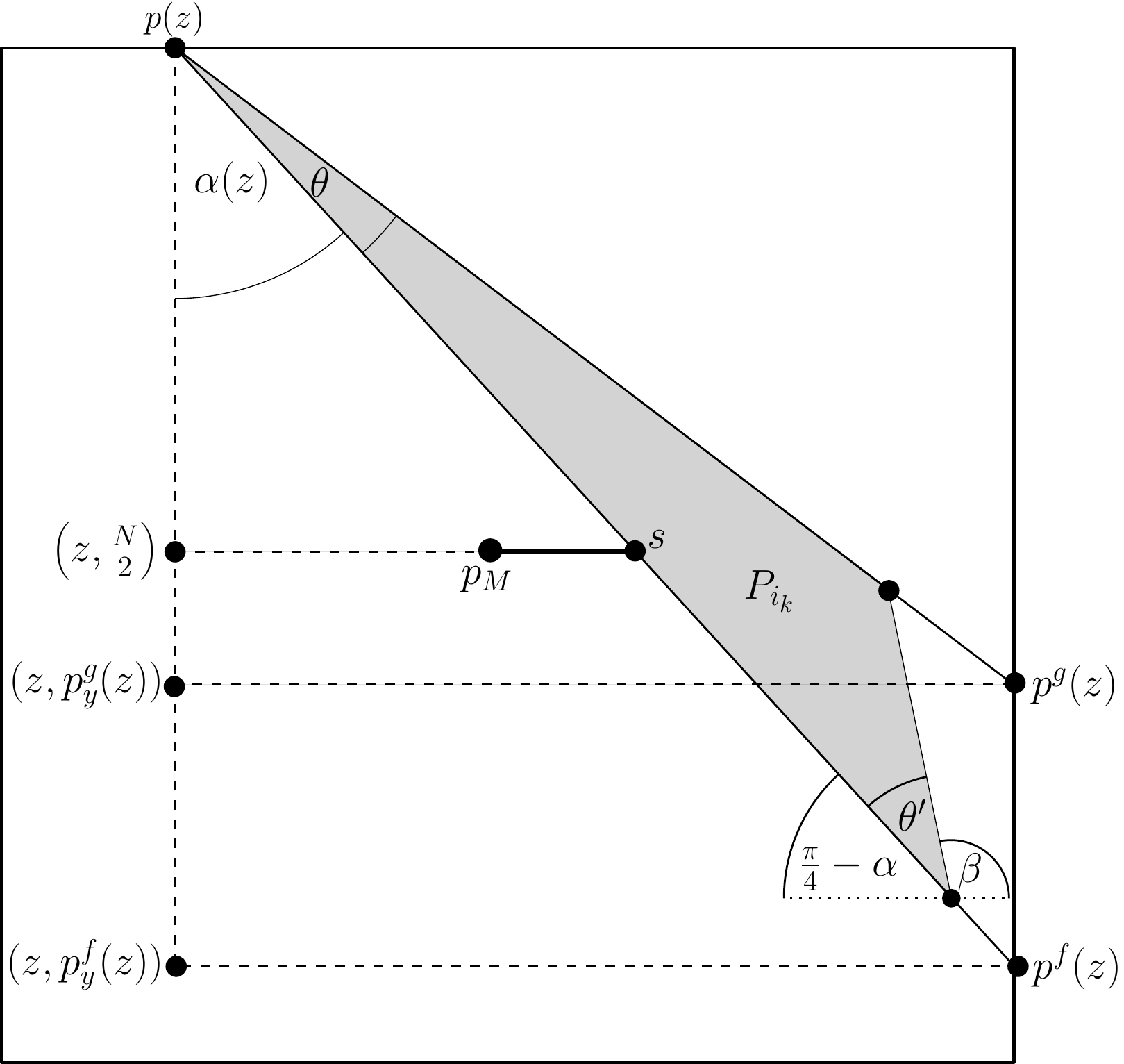}
	\caption{The points, line segments, and angles used in the proof of Lemma~\ref{lem:right-facing-tlp}}
	\label{fig:points_lemma} 
\end{figure}

\begin{lem}\label{lem:right-facing-tlp}
If each triangle in $\OPT_{W}$ is right-facing we have that $\OPT_{W}^{(t)}\cup\{P_{i_{k}}\}$
is top-left-packable.
\end{lem}

\begin{proof}
Let $s^{*}=\sum_{P_{i_{\ell}}\in\OPT_{W}^{(t)}}\hat{d}_{i_{\ell}}$ and $s=(s^{*}+N/2,N/2)$.
Consider now $p(z)=(z,N)$ and define the point $p^{f}(z)=(p_{x}^{f}(z),p_{y}^{f}(z))$
as the intersection between the right side of the Knapsack and $\overline{p(z)s}$.
Similarly, define $p^{g}(z)=(p_{x}^{g}(z),p_{y}^{g}(z))$ as the intersection
between $\{N\}\times \R$ and the line $L^{\theta}$
obtained by rotating $\overline{p(z)s}$ around $p(z)$ by $\theta:=\theta_{i_k}$ counterclockwise. Note that every polygon in $\OPT_{W}^{(t)}$ is contained in $H^- := K \cap\conv(\{ p_{TL},p^f(0),p_{BR},p_{BL}\})$
as they are top-left-packed. 

We translate $P_{i_{k}}$ upwards until it intersects the top side of the knapsack at a point $(z^{*},N)$. By Lemma \ref{lem:close_top} we get that $z^* \in [0,2\epsilon N]$. Note that now $P_{i_{k}}$ is placed inside the triangle $p(z^{*}),p^{f}(z^{*}),p^{g}(z^{*})$.
Define now $f(z):=\|p(z)-p^{f}(z)\|^2$, $g(z):=\|p(z)-p^{g}(z)\|^2$.
Note that it suffices to prove that $f(0)\geq f(t)$ and $g(0)\geq g(t)$
for each $t\in[0,2\epsilon N]$ as this implies that $P_{i_{k}}$ can be
placed inside the triangle with vertices $p(0),p^{f}(0),p^{g}(0)$, this triangle is contained in $K\setminus H^-$ and this placement is group respecting.

Since $\frac{N}{2}+s^*-z \geq \left( \frac{1}{2}-2\epsilon\right)N>0$ and $N-z\geq (1-2\epsilon)N>0$, a similarity argument between triangles $p(z),(z,p_{y}^{f}(z)),p^{f}(z)$
and $p(z),(z,N/2),s$ lets us obtain: 
\begin{equation}
\frac{\frac{N}{2}}{\frac{N}{2}+s^{*}-z}=\frac{N-p_{y}^{f}(z)}{N-z}.\label{eq:sim1}
\end{equation}
As $p_y^f(z)\geq 0$ we obtain $N-z\leq  N+2s^*-2z$, implying that $z \leq 2s^*$ and therefore we only need to prove that $f(0)\geq f(z)$ for each $z\in [0, \min\{\epsilon N , 2s^*\}]$. Pythagoras theorem on $p(z),p^{f}(z),p_{TR}$ gives us: 
\begin{equation}
f(z)=(N-z)^{2}+(N-p_{y}^{f}(z))^{2}\label{eq:pyt1}
\end{equation}
Combining \ref{eq:sim1} and \ref{eq:pyt1} we get: 
\begin{align*}
	f(z)& = (N-z)^2 + (N-p_{y}^{f}(z))^{2} \\
	& = (N-z)^ 2 \left( 1 + \frac{N^2}{(N+2s^*-2z)^2}  \right) 
\end{align*}
We now compute $f'(z)$.
\begin{align*}
f'(z) & =  -2(N-z) \left( 1 + \frac{N^2}{(N+2s^*-2z)^2}  \right)+(N-z)^ 2 \left( \frac{4N^2}{(N+2s^*-2z)^3}  \right) \\
& = \frac{(N-z)}{(N+2s^*-2z)^3} \left(-2(N+2s^*-2z)^3 -2N^2(N+2s^*-2z)+4N^2(N-z) \right)\end{align*}

Note that $f'(z)=0$ if and only if $\phi(z):=4N^2(N-z)-2(N+2s^*-2z)^3 -2N^2(N+2s^*-2z) =0$. After some algebra we obtain the following:
\[ \phi(z) = z^3\cdot 16+z^2(-48s^*-24N)+z(48(s^*)^2+48s^*N+12N^2)-16 N^2s ^*- 24 N (s^*)^2 - 16 (s^*)^3  \]

We can compute the discriminant of this polynomial and obtain:
\[\Delta_\phi = -27684 (N^6 -4N^5s^*+4N^4(s^*)^2)=-27684N^4 (N-2s^*)^2 < 0 \]
since $s^* < \frac{N}{2}$ (as $\OPT_W$ intersects $L_1$ by at least $s^*$ and the amount $P_{i_k}$ intersects $L_1$).  Therefore $\phi$ has a unique real root $\overline{z}$. Note that $\max\limits_{z \in [0,2s^*]} f(z) \in \{f(0), f(2s^*), f(\overline z))\}$. We begin by showing that $f(0)\geq f(2s^*)$. Indeed:

\begin{align*}
f(0)-f(2s^*) & = N^ 2 \left( 1 + \frac{N^2}{(N+2s^*)^2}  \right) -(N-2s^*)^ 2 \left( 1 + \frac{N^2}{(N-2s^*)^2}  \right) \\
& = \frac{N^4-(N+2s^*)^2(N-2s^*)^2}{(N+2s^*)} \\
& =  \frac{N^4-(N^2-(2s^*)^2)^2}{(N+2s^*)} \geq 0.
\end{align*} 

Therefore $\max\limits_{z \in [0,2s^*]} f(z) \in \{f(0), f(\overline z))\}$. If $\overline z \notin [0,2s^*]$, then it is clear that $f(0)=\max\limits_{z \in [0,2s^*]} f(z)$. Suppose then that $\overline{z}\in [0,2s^*]$. Since:

\[f'(0)= \frac{N}{(N+2s^*)^3} \left(-2(N+2s^*)^3 -2N^2(N+2s^*)+4N^3 \right) \leq 0,\]
then $f'(z)\leq 0$ for $z \in [0,\overline{z}]$ by continuity. Therefore $f$ is decreasing in $[0,\overline z]$ concluding that $f(0)\geq f(\overline z)$ and that $f(0)\geq f(z)$ for $z \in [0,\min\{2\epsilon N , 2s^*\}]$.

Let $\alpha:=\alpha(z)$ be the angle between $L$ and
$\overline{p(z)p^{f}(z)}$. Let $u(z)=(u_x(z),u_y(z))$ be the placement of the vertex of $P_{i_k}$ that is not adjacent to the longest edge. We aim to show that $u(0) \in K$ implying that $P_{i_k}$ is placed inside the Knapsack. By examining the triangle $p(0)$, $(z,p_y^g(0))$, $p^g(0)$ we get that $\tan (\alpha(0)+\theta)= \frac{\|p^g(0)-(0,p^g_y(z))\|}{\|p(0)-(0,p^g_y(z))\|}=\frac{N}{|N-p^g(0)|}$. 

Furthermore by applying the law of sines on triangle $p(0)$, $p^f(0)$, $p^g(0)$ we get:
\[ \frac{\sin \theta }{\|p^g(0)-p^f(0)\|}=\frac{\sin \alpha(0) }{\|p^g(0)-p(0)\|} \]

Therefore $\|p^g(0)-p^f(0)\| = \frac{\|p^g(0)-p(0)\|\sin\theta }{\sin\alpha(0)} \leq \frac{\sin\theta \ell_i}{\sin\alpha(0)} \leq \frac{\epsilon\sqrt 2 }{\frac{\sqrt 2}{2}-\epsilon} N$. By choosing $\epsilon$ small enugh we assume that 
$\|p^g(0)-p^f(0)\| \leq \frac{1}{10}N$. 
Note that by choosing $\epsilon$ small enough $\tan(\alpha(0)+\theta)\leq (1+\frac{1}{10})$.
Therefore:
\[ N \leq \left(1+\frac{1}{10}\right)|N-p^g(0)| \]
Which implies that $p^g_y(0) \leq \frac{1}{10}N$
 or $p^g_y(0) \geq \frac{21}{10}N$. If $p^g_y(0) \geq \frac{21}{10}N$ then $p^f(0)_y \geq \frac{20}{10}N$ which is not possible. We conclude that $p^g_y(0)\leq \frac{1}{10}N$ and since $u(0) \in \conv(\{ p^g(0), p(0) \}$ we conclude that $0 \leq u_y(0)\leq N$. It only remains to prove that the same holds for $u_x(0)$. 
 
 Let $u'(z)=(u'_x(z),u'_y(z))$ be the vertex of $P_{i_k}$ that is not $u$ or $v_{i_k}^*$. Let $\theta'$ be the angle of $P_{i_k}$ at $u'$ and $\beta(z)$ the angle between $\overline{uu'}$ and $\{u'_y(z)\} \times [u_x'(z),\infty)$. Note that $\beta(z)+\theta' + \frac{\pi}{2}-\alpha(z) = \pi$,  therefore $\alpha'(z)=\beta'(z)$. By examining the triangle $p(z)$, $s$, $(z,N/2)$ we obtain that $\tan (\alpha (z)) = \frac{\|s-(z,N/2)\|}{\|p(z)-(z,N/2)\|}= \frac{N+2s^*-2z}{N}= 1 +2\frac{s^*-z}{N}$. Therefore:
 
 \[\beta'(z)=\alpha'(z)= \frac{-\frac{2}{N}}{1+\left( 1 +2\frac{s^*-z}{N}\right)^2} \leq 0 \]
 
 We conclude that $\beta(0)\geq \beta(z)$ for each other $z$. Call $\ell'= \|u'-u\|$ and $R(\beta)$ the rotation matrix by $\beta$, then: $u(z)=u'(z)+ R(\beta (z))\binom{0}{\ell'}$.
 
 In particular $u_y(0)=u'(0)- \sin(\beta(0))\ell' \leq u'(z^*)- \sin(\beta(z^*))\ell' \leq u_y(z^*)\leq N$, from which we conclude.
 \end{proof}

\begin{lem}
If each triangle in $\OPT_{W}$ is bottom-facing we have that $\OPT_{W}^{(t)}\cup\{P_{i_{k}}\}$
is bottom-right-packable.
\end{lem}
\begin{proof}
	Let $s^* = \sum_{P_{i_\ell}\in \OPT_W^{t}} \hat d_{i_\ell}$, $s^r=(s^*+\frac{N}{2},\frac{N}{2} )$ and $s^\ell=(s^*+\frac{N}{2},\frac{N}{2})$. We define $p(z)$ as $(z,N)$. Define $p^f(z)$ as the intersection between $\overline{p(z)s^r}$ and the bottom edge of $K$. Similarly, let $L^\theta$ be the rotation of $\overline{p(z)s^r}$ by $\theta$ counterclockwise and call $p^ g(z)$ the intersection between $L^\theta$ and the bottom edge. Let also $T(z)$ the triangle $p(z)$, $p^f(z)$, $p^g(z)$. We begin by rotating $\OPT_W^{(t)}$ by 180\degree~around $p_M$.
	
	We now translate $P_{i_k}$ upwards until it intersects the top-edge of the Knapsack at $p(z^*)$ for some $z^*$. Note that $P_{i_k}$ is contained inside the triangle $T(z^*)$ 
	
	Let $I$ be the amount $T(z^*)$ intersects $L$. A similarity argument between $T(z^*)$ and $(z^*,N)$, $(z^*, N/2)$, $s^r +(I,0)$ gives us:
	\[ \frac{1}{2} = \frac{N/2+s^* -z^*+I}{p^g(z)-z^*} \]
	Since $p^g(z^*)\leq N$ we obtain $2s+2I\leq z^*$.
	
	We now translate $P_{i_k}$ to the left until $v_{i_k}^*$ coincides with $p_{TL}$. Let $q=(q_x,q_y)$ be the rightmost point in $P_{i_k} \cap L$. Since $2s+2I\leq z^*$ we know that $q_x \leq \frac{N}{2}-s$. Therefore $P_{i_k}$ is placed to the left of the line $L^*$ that passes through $v_{TL}$ and $s^\ell$. Furthermore, $\OPT_W^ {(t)}$ is to the right of $L^*$ as they are bottom-right packed. We know rotate $P_{i_k}$ counterclockwise around $v^*_{i_k}$ until it overlaps $s^\ell$. Finally, by rotating  $\OPT_W^{(t)}\cup \{P_{i_k}\}$ 180\degree~ around $p_M$ we arrive at a bottom-right packing of $\OPT_W^{(t)}\cup \{P_{i_k}\}$.
\end{proof}

We add $P_{i_{k}}$ to $\OPT_{W}^{(t)}$. We continue iteratively
until we assigned all triangles in $\OPT_{W}$ to the sets $\OPT_{W}^{(1)},...,\OPT_{W}^{(g)}$.
Then the most profitable set $\OPT_{W}^{(t^{*})}$ among them satisfies
that $w(\OPT_{W}^{(t^{*})})\ge\frac{1}{g}w(\OPT_{W})$. On the other
hand, $w(P_{i^{*}})\ge \Omega(w\left(\OPT\cap\Ph\setminus\OPT_{W}\right))$.
Hence, $w(\OPT_{W}^{(t^{*})})\ge\frac{1}{g}w(\OPT\cap\Ph)$ or $w(P_{i^{*}})\ge\Omega(w\left(\OPT\cap\Ph\right))$
which completes the proof of Lemma~\ref{lem:exists-good-packable-solution}.
\subsection{Hard polygons under resource augmentation}
Let $\delta>0$. We consider the setting of $(1+\delta)$-resource
augmentation, i.e., we want to compute a solution $\P'\subseteq\P$
that is feasible for a knapsack of size $(1+\delta)N\times(1+\delta)N$
and such that $w(\OPT)\le O(w(\P'))$ where $\OPT$ is the optimal
solution for the original knapsack of size $N\times N$. Note that
increasing $K$ by a factor of $1+\delta$ is equivalent to shrinking
the input polygons by a factor of $1+\delta$.

Given a polygon $P$ defined via coordinates $(x_{1},y_{1}),...,(x_{k},y_{k})\in\R^{2}$
we define $\shr_{1+\delta}(P)$ to be the polygon with coordinates
$(\bar{x}_{1},\bar{y}_{1}),...,(\bar{x}_{k},\bar{y}_{k})\in\R^{2}$
where $\bar{x}_{k'}=x_{k'}/(1+\delta)$ and $\bar{y}_{k'}=y_{k'}/(1+\delta)$
for each $k'$. For each input polygon $P_{i}\in\P$ we define its
shrunk counterpart to be $\bar{P}_{i}:=\shr_{1+\delta}(P_{i})$. 
{} Based on $\bar{\P}$ we define sets $\bPe,\bPm,\bPh$ and the set
$\bar{\P}_{j}$ for each $j\in\Z$ in the same way as we defined $\Pe,\Pm,\Ph$
and $\P_{j}$ based on $\P$ above.

For the sets $\bPe$ and $\bPm$ we use the algorithms due to Lemmas~\ref{lem:easy-polygons}
and \ref{lem:medium-polygons} as before. For the hard polygons $\bPh$
we can show that there are only $O_{\delta}(1)$ groups $\bar{\P}_{j}$
that are non-empty, using that we obtained them via shrinking the
original input polygons. Intuitively, this is true since $\bar{\ell}_{i}\leq\frac{\sqrt{2}N}{1+\delta}$
for each $\bar{P}_{i}\in\bar{\P}$ where $\bar{\ell}_{i}$ denotes
the length of the longest diagonal of $\bar{P}_{i}$, and hence $\bar{\P}_{j}\cap\bPh=\emptyset$
if $j<\log\left(\frac{\delta}{1+\delta}\sqrt{2}N\right)$. 
\begin{lem}
\label{lem:few-groups-shrink-1}We have that $\bar{\P}_{j}=\emptyset$
if $j<\log\left(\frac{\delta}{1+\delta}\sqrt{2}N\right)$. Hence,
there are only $\log\left(\frac{1+\delta}{\delta}\right)+1$ values
$j\in\Z$ such that $\bar{\P}_{j}\ne\emptyset$. 
\end{lem}
\begin{proof}
Let $P_{i}$ be an arbitrary polygon. Note that $\bar{\ell}_{i}=\frac{1}{1+\delta}\ell_{i}$,
and therefore $\bar{\ell}_{i}\leq\frac{\sqrt{2}N}{1+\delta}$. We
conclude that $\frac{\delta}{1+\delta}\sqrt{2}N\leq\sqrt{2}N-\bar{\ell}_{i}\leq\sqrt{2}N$.
Note that, for any $j$, if $\bar{\P}_{j}$ is non-empty, there must
be a $P_{i}$ that satisfies $\sqrt{2}N-2^{j}\leq\bar{\ell}_{i}<\sqrt{2}N-2^{j-1}$
(or equivalently $2^{j-1}<\sqrt{2}N-\bar{\ell}_{i}\leq2^{j}$). We
conclude that such $j$'s must verify $\log\left(\frac{\delta}{1+\delta}\sqrt{2}N\right)\leq j<\log(2\sqrt{2}N)$
and therefore there are at most $\log(2\sqrt{2}N)-\log\left(\frac{\delta}{1+\delta}\sqrt{2}N\right)=\log\left(\frac{1+\delta}{\delta}\right)+1$
non-empty $\bar{\P}_{j}$
\end{proof}
Lemmas~\ref{lem:Pj-properties} and \ref{lem:few-groups-shrink-1}
imply that $|\overline{\OPT}\cap\bPh|\le O(\log\left(\frac{1+\delta}{\delta}\right))$
where $\overline{\OPT}$ denotes the optimal solution for the polygons
in $\bar{\P}$. Let $\bPh'\subseteq \bPh$ denote the set due to Lemma~\ref{lem:placement-hard-polygons}
when assuming that $\bPh$ are the hard polygons in the given instance.
Therefore, we guess $\bPh'$ in
time $n^{O(\log\left(\frac{1+\delta}{\delta}\right))}$. 
Finally, we output the solution of largest weight
among $\bPh'$ and the solutions due applying to Lemmas~\ref{lem:easy-polygons}
and \ref{lem:medium-polygons} to the input sets $\bPe$ and $\bPm$,
respectively. This yields the proof of Theorem~\ref{thm:RA-ptime}.

\section{Optimal profit under resource augmentation}

In this section we also study the setting of $(1+\delta)$-resource augmentation,
i.e., we want to compute a solution $\P'$ which is feasible for an
enlarged knapsack of size $(1+\delta)N\times(1+\delta)N$, for any
constant $\delta>0$. We present an algorithm with a running
time of $n^{(\log(n)/\delta)^{O(1)}}$ that computes such a solution
$\P'$ with $w(\P')\ge\OPT$ where $\OPT$ is the optimal solution
for the original knapsack of size $N\times N$. In particular, we
here do not lose any factor in our approximation guarantee.

First, we prove a set of properties that we can assume ``by $(1+\delta)$-resource
augmentation'' meaning that if we increase the size of $K$ by a
factor $1+\delta$ then there exists a solution of weight $w(\OPT)$
with the mentioned properties, or that we can modify the input in
time $n^{O(1)}$ such that it has these properties and there still
exists a solution of weight $w(\OPT)$.

\subsection{Few types of items}

We want to establish that the input polygons have only $(\log(n)/\delta)^{O(1)}$ different shapes.
Like in Section~\ref{sec:Constant-factor-approximation} for each
polygon $P_{i}\in\P$ denote by $B_{i}$ its bounding box with width
$\ell_{i}$ and height $h_{i}$. Note that $\ell_{i}\ge h_{i}$. The
bounding boxes of all polygons $P_{i}\in\P$ such that $h_{i}\le\delta\frac{N}{n}$
have a total height of at most $\delta N$. Therefore, we can pack
all these polygons into the extra capacity that we gain by increasing
the size of $K$ by a factor $1+\delta$ and therefore ignore them
in the sequel. 
\begin{lem}
\label{lem:large-height-width}By $(1+\delta)$-resource augmentation
we can assume for each $P_{i}\in\P$ that $\ell_{i}\ge h_{i}\ge\delta N/n$
and that $\area(P_{i})=\Omega(\area(K)\delta^{2}/n^{2})$. 
\end{lem}
\begin{proof}
Let $v_1,v_2$ be the two vertices that define $h_i$. Note that $d(v_1,v_2)\geq h_i$ by the triangle inequality and $d(v_1,v_2)\leq \ell_i$ by definition, concluding that $\ell_i \geq h_i$. Additionally, let $\mathcal P_{thin}$ be the polygons $P_i$ that satisfy $h_i\leq \delta \frac{N}{n}$. Note that $\sum_{P_i \in \mathcal \P_{thin}} h_i \leq \delta N$, and therefore, we can pack these polygons into the extra capacity. Finally, by Lemma~\ref{lem:area-bound} $\area (P_i)\geq \frac{1}{2}\ell_{i}h_i = \frac{1}{2}\frac{\delta ^ 2 N^ 2}{n^2 }$ for each remaining polygon $P_i$.
\end{proof}
Next, intuitively we stretch the optimal solution $\OPT$ by a factor
$1+\delta$ which yields a container $C_{i}$ for each polygon $P_{i}\in\OPT$
which contains $P_{i}$ and which is slightly bigger than $P_{i}$.
We define a polygon $P'_{i}$ such that $P_{i}\subseteq P'_{i}\subseteq C_{i}$
and that globally there are only $(\log(n)/\delta)^{O_{\delta}(1)}$
different ways $P'_{i}$ can look like, up to translations and rotations.
We refer to those as a set $\S$ of \emph{shapes} of input objects.
Hence, due to the resource augmentation we can replace each input
polygon $P_{i}$ by one of the shapes in $\S$. 
\begin{lem}
\label{lem:few-items} By $(1+\delta)$-resource augmentation we can
assume that there is a set of shapes $\S$ with $|\S|\le(\log(n)/\delta)^{O_{\delta}(1)}$
such that for each $P_{i}\in\P$ there is a shape $S\in\S$ such
that $P_{i}=S$ and $S$ has only $\Lambda={(1/\delta)}^{O(1)}$ many
vertices. 
\end{lem}
\begin{proof}
Let $P$ be an input polygon. Assume that $P$ is rotated such that
its longest diagonal is horizontal, i.e., let $p=(p_{x},p_{y}),q=(q_{x},q_{y})$
denote the vertices of $P$ with largest distance; we assume that
$p$ and $q$ lie on a horizontal line, i.e., that $p_{y}=q_{y}$.
Furthermore, let $r=(r_{x},r_{y}),s=(s_{x},s_{y})$ denote the vertices
of $P$ with minimum and maximum $y$-coordinate, respectively (see Figure \ref{fig:lem25}).

\begin{figure}[t]
	\begin{centering}
		\includegraphics[scale=0.6]{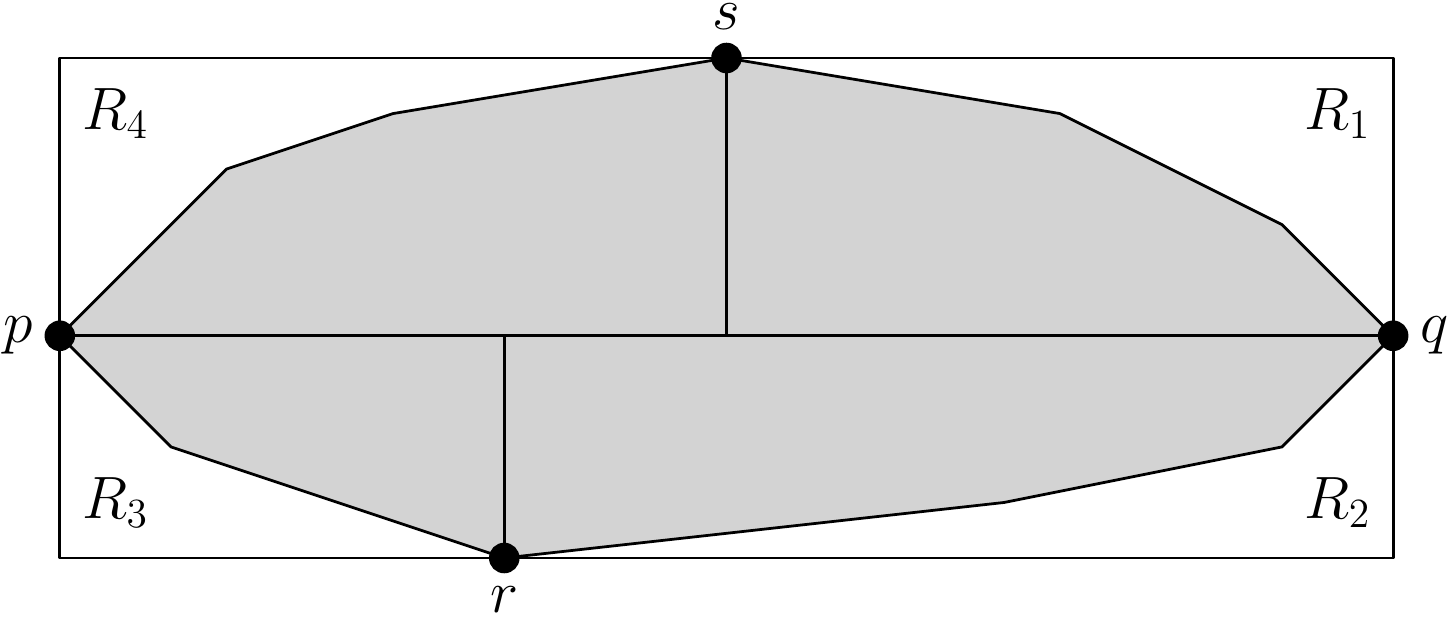}
	\end{centering}
	\caption{\label{fig:lem25} The points and rectangles used in the proof of Lemma \ref{lem:few-items}.}
\end{figure}

Note that extending the knapsack by a factor $1+\delta$ is equivalent
to shrinking each input polygon by a factor $1+\delta$. Let $\shr_{1+\delta}(P)$
denote the polygon obtained by shrinking $P$ towards the origin,
i.e., by replacing each vertex $v=(v_{x},v_{y})$ of $P$ by the vertex
$v'=(\frac{v_{x}}{1+\delta},\frac{v_{y}}{1+\delta})$. Our goal is
to show that there exists a polygon $P'$ whose shape is one shape
out of $(\log(n)/\delta)^{O_{\delta}(1)}$ options such that there
is a translation vector $\overrightarrow{a}$ with $\overrightarrow{a}+\shr_{1+\delta}(P)\subseteq P'\subseteq P$.
Then we replace in the input the polygon $P$ by a polygon $\tilde{P}$
which is congruent to $P'$. Hence, for the shapes of the resulting
polygons $\tilde{P}$ there are only $(\log(n)/\delta)^{O_{\delta}(1)}$
options. In the process, we will shrink $P$ a constant number of
times. Then the claim follows by redefining $\delta$ accordingly.

First, we shrink $P$ by a factor of at most $1+\delta$ such that
the line segment connecting $p$ and $q$ has a length that is a power
of $1+\delta$. Let $\ell'_{i}$ denote this new length. Since originally
$\sqrt{2}N\ge\ell_{i}\ge\delta N/n$ there are only $O(\log_{1+\delta}n)$
options for $\ell'_{i}$. We partition the bounding box of $P$ into
four rectangles where
\begin{itemize}
\item $R_{1}$ is the (unique) rectangle with vertices $s$ and $q$,
\item $R_{2}$ is the (unique) rectangle with vertices $r$ and $q$,
\item $R_{3}$ is the (unique) rectangle with vertices $p$ and $r$, and
\item $R_{4}$ is the (unique) rectangle with vertices $p$ and $s$.
\end{itemize}
We translate $P$ such that $p$ is the origin. If the width of $R_{1}$
is smaller than $\delta\ell'_{i}$ then intuitively we shrink $P$
by a factor $1+\delta$ towards $p$ such that $q_{x}$ is again a
power of $1+\delta$ and $s_{x}=q_{x}$. More formally, we replace
each vertex $v=(v_{x},v_{y})$ of $P$ by a vertex $v'=(v'_{x},v'_{y})$
such that $v'_{x}=\alpha v_{x}$ and $v'_{y}=\alpha v_{y}$ for some
value $\alpha\in(\frac{1}{1+\delta},1]$. First, we move $q$ towards
$p$ such that $q_{x}$ is the next smaller power of $1+\delta$.
Then we move $s$ towards $p$ such that $s_{x}=q_{x}$. Finally,
we move each remaining vertex $v$ by exactly a factor $1+\delta$
towards $p$. As a result, $R_{1}$ becomes empty. We perform similar
operations in case that the width of $R_{2},R_{3}$, or $R_{4}$ is
smaller than $\delta\ell'_{i}$. Also, we perform a similar operation
in case that the height of $R_{1}$ (identical to the height of $R_{4}$)
is smaller than $\delta h_{i}$ or that the height of $R_{2}$ (identical
to the height of $R_{3}$) is smaller than $\delta h_{i}$. In the
latter operations we move the vertices of $P$ towards $s$ or $r$,
respectively.

Assume again that $p$ is the origin. Let $t:=(s_{x},p_{y})$. Let
$t'=(t'_{x},t'_{y})$ such that $t'_{y}=p_{y}$, and $t'_{x}$ is
the smallest value $t'_{x}$ with $t'_{x}\ge t_{x}$ such that the
distance between $t'$ and $q$ is a multiple of $\delta^{3}\ell'_{i}$.
In particular, then $t'_{x}-t_{x}\le\delta^{3}\ell'_{i}\le\delta^{2}(q_{x}-t_{x})$
and note that $q_{x}-t_{x}$ is the width of $R_{1}$ for which $q_{x}-t_{x}\ge\delta\ell'_{i}$
holds.

We define $s'=(s'_{x},s'_{y})$ such that $s'_{x}=t'_{x}$ and $s'_{y}$
is that largest value $s'_{y}$ such that $s'=(s'_{x},s'_{y})$ lies
inside $P$. Observe that $s'_{y}\ge(1-\delta^{2})s_{y}$ since $P$
includes all points on the line segment connecting $s$ and $q$ by
convexity. Similarly, we define a point $t''$ between $p$ and $t$
and a corresponding point $s''$. We move each vertex $v=(v_{x},v_{y})$
of $P$ towards $t$ that satisfy that $t''_{x}\le v_{x}\le t'_{x}$
and $v_{y}\ge t''_{y}=t_{y}=t'_{y}$, i.e., we reduce the distance
between $v$ and $t$ by a factor $1/(1+\delta)$ which we justify
via shrinking. One can show that afterwards $v$ lies in the convex
hull spanned by the other vertices of $P$ and $s'$ and $s''$, using
that $s'_{y}\ge(1-\delta^{2})s_{y}\ge(1-\delta^{2})v_{y}$. Hence,
we can remove $v$.

We move $P$ such that $t'$ becomes the origin. Let $R'_{1}$ denote
the (unique) rectangle with vertices $s'$ and $q$. Our goal is
now to move the vertices within $R'_{1}$ such that only $O_{\delta}(1)$
vertices remain and that for the coordinate of each of them there
are only $(\log(n)/\delta)^{O_{\delta}(1)}$ options. Whenever we
move a vertex $v$ within $R'_{1}$ we move $v$ towards $t'$ such
that the distance between $v$ and $t'$ decreases by at most a factor
$1+\delta$ but keep $s'$ and $q$ unchanged. In this way the distance
between $t'$ and $q$ does not change and the distance between $t'$
and $s'$ does not change either. Let $h$ denote the distance between
$t'$ and $s'$ and let $w$ denote the distance between $t'$ and
$q$. Assume w.l.o.g.~that $h=w=1$ and that $t'$ is the origin.
Observe that by convexity each point on the line segment connecting
$s'$ and $q$ lies within $P$.

Let $k\in\N$ be a constant with $k=O_{\delta}(1)$ to be defined
later. We shoot rays $r_{0},...,r_{k}$ originating at $t'$ such
that $r_{0}$ goes through $s'$, $r_{k}$ goes through $q$, and
between any two consecutive rays $r_{j},r_{j+1}$ there is an angle
of exactly $\frac{\pi}{2k}$ (see Figure~\ref{fig:shrunk}).
\begin{figure}[t]
	\begin{centering}
		\includegraphics[scale=0.6]{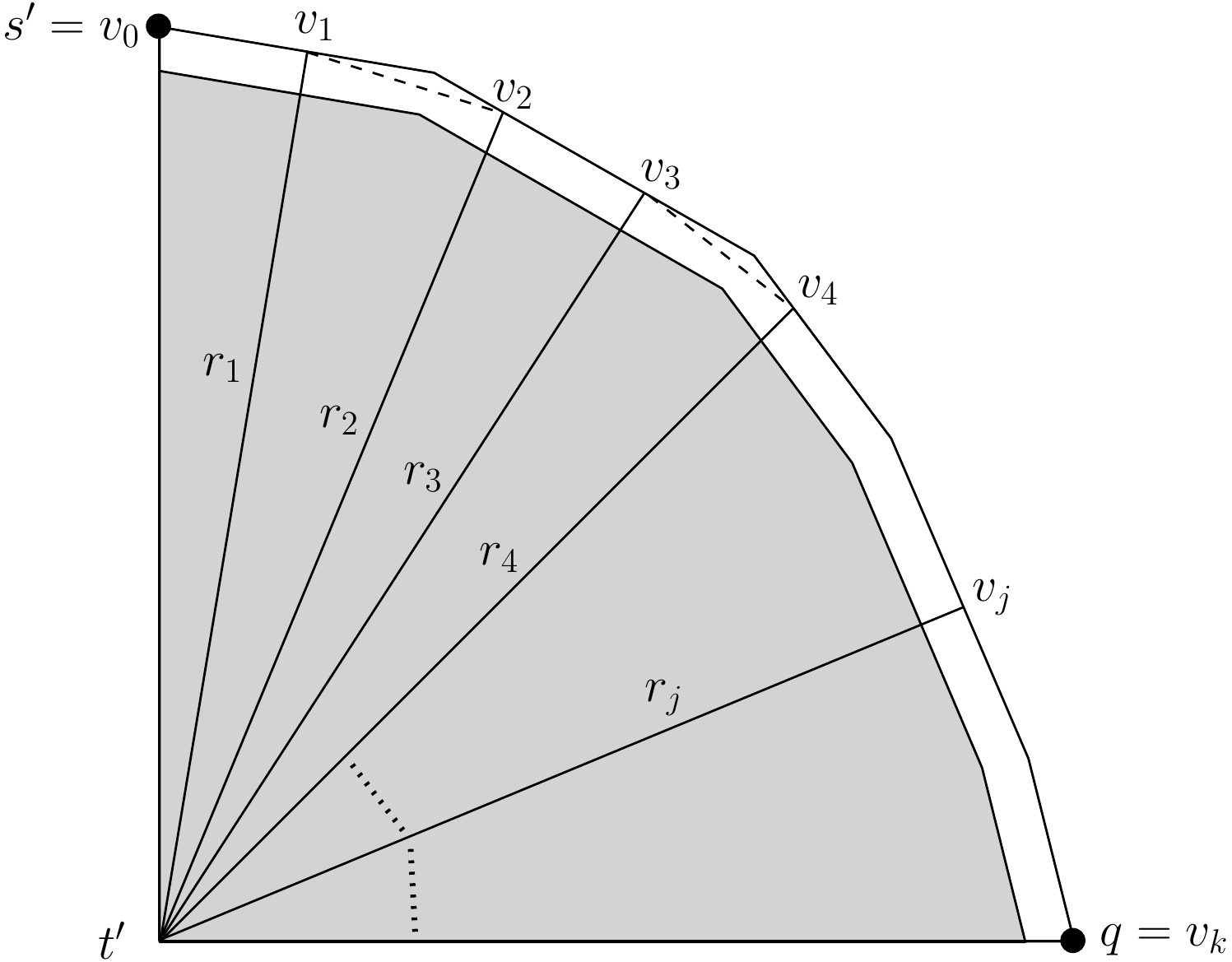}
	\end{centering}
	\caption{\label{fig:shrunk} The points and rays related to the shrinking of $R_1$. In gray we see the shrunk polygon and in white the original.}
\end{figure}
For each $j\in\{0,...,k\}$
denote by $v_{j}$ the point on the boundary of $P$ that is intersected
by $r_{j}$ (not necessarily a vertex of $P$). Imagine that we shrink
$P$ such that we move each vertex $v=(v_{x},v_{y})$ of $P$ towards
$t'=0$, i.e., we replace $v$ by the point $v':=(\frac{v_{x}}{1+\delta},\frac{v_{y}}{1+\delta})$.
We argue that $v'$ lies in the convex hull of $t',v_{0},...,v_{k}$
and hence we can remove $v'$. Let $r$ be a ray originating at $t'$
and going through $v$. Suppose that $r_{j}$ and $r_{j+1}$ are the
rays closest to $r$. Let $v_{j}=(v_{j}^{x},v_{j}^{y})$ and $v_{j+1}=(v_{j+1}^{x},v_{j+1}^{y})$.
Then $v$ lies in the convex hull of $v_{j},v_{j+1}$ and the point
$(v_{j+1}^{x},v_{j}^{y})$. We have that $v_{j}^{x}\ge1/3$ and $v_{j+1}^{x}\ge1/3$
or that $v_{j}^{y}\ge1/3$ and $v_{j+1}^{y}\ge1/3$. Assume w.l.o.g.~that
$v_{j}^{x}\ge1/3$ and $v_{j+1}^{x}\ge1/3$. We claim that then $v_{j}^{x}\le v_{x}\le v_{j+1}^{x}\le(1+\delta)v_{j}^{x}$.
The first two inequalities follow from convexity. For proving that
$v_{j+1}^{x}\le(1+\delta)v_{j}^{x}$, we can assume that $v_{j}^{x}\le1/(1+\delta)$
since otherwise the claim is immediate. This implies that $v_{j}^{y}\ge\Omega(\delta)$.
Also observe that $v_{j+1}^{y}\le v_{j}^{y}$ since otherwise $v_{j}$
would not be on the boundary of $P$, by convexity. Therefore, $v_{j+1}^{x}$
cannot be larger than the $x$-coordinate of the point on $r_{j+1}$
with $y$-coordinate $v_{j}^{y}$. Using that $1/(1+\delta)\ge v_{j}^{x}\ge1/3$
one can show that there is a choice for $k\in O_{\delta}(1)$ that
ensures that $v_{j+1}^{x}\le(1+\delta)v_{j}^{x}$. Finally, for each
point $v_{j}$ with $j\in\{1,...,k-1\}$ we move $v_{j}$ towards
$t'$ such that the distance between $v_{j}$ and $t'$ becomes a
power of $1+\delta$. Since before the shrinking this distance $\Omega(1)$
there are only $O_{\delta}(1)$ options for the resulting distance.

In a similar way we define $R'_{2},R'_{3}$, and $R'_{4}$ and perform
a symmetric operation on them. The resulting polygon is defined via
$\ell'_{i}$, the positions of $t',t''$, the positions of the vertices
$u'$ and $u''$ (which are defined analogously to $t'$ and $t''$),
for $R'_{1}$ the distance between $t'$ and $s'$ and the distances
of the $O_{\delta}(1)$ vertices $v_{j}$ to $t'$, and the respective
values for $R'_{2},R'_{3}$, and $R'_{4}$. For each of these values
there are only $(\log(n)/\delta)^{O_{\delta}(1)}$ options and there
are $O_{\delta}(1)$ such values in total. Hence, there are $(\log(n)/\delta)^{O_{\delta}(1)}$
possibilities for the resulting shape.
\end{proof}
Finally, we ensure that for each polygon $P_{i}\in\P$ we can restrict
ourselves to only $(n/\delta)^{O(1)}$ possible placements in $K$.
\begin{lem}
\label{lem:few-positions} By $(1+\delta)$-resource augmentation,
for each polygon $P_{i}\in\P$ we can compute a set $\L_{i}$ of at
most $(n/\delta)^{O(1)}$ possible placements for $P_{i}$ in time
$(n/\delta)^{O(1)}$ such that if $P_{i}\in\OPT$ then in $\OPT$
the polygon $P_{i}$ is placed inside $K$ according to one placement
$\tilde{P}_{i}\in\L_{i}$.
\end{lem}
\begin{proof}
First, we prove that for each polygon $P_{i}$ it suffices to allow
	only $(n/\delta)^{O(1)}$ possible vectors $d$ when defining its
	placement $\tilde{P}_{i}$ as $\tilde{P}_{i}=d+\rot_{\alpha}(P_{i})$.
	To this end, assume that $\OPT=\{P_{1},...,P_{k}\}$. For each polygon
	$P_{i}\in\OPT$ denote by $\tilde{P}_{i}$ its corresponding placement
	in $\OPT$. We assume that for any $\tilde{P}_{i},\tilde{P}_{i'}$
	with $i<i'$ it holds that intuitively $\tilde{P}_{i}$ lies on the
	left of $\tilde{P}_{i'}$. Formally, we require that if there is a
	horizontal line $L$ that has non-empty intersection with both $\tilde{P}_{i}$
	and $\tilde{P}_{i'}$ then $L\cap\tilde{P}_{i}$ lies on the left
	of $L\cap\tilde{P}_{i'}$. Since the polygons are convex such an ordering
	exists.
	
	Now for each $k'\in\{1,...,k\}$ we move $\tilde{P}_{k'}$ by $k'\cdot\frac{\delta}{n}N$
	units to the right. Since $k\le n$ the resulting placement fits into
	the knapsack using $(1+\delta)$-resource augmentation. Intuitively,
	in the resulting placement, each polygon $\tilde{P}_{i}$ has $\frac{\delta}{n}N$
	units of empty space on its left and on its right.
	
	In a similar fashion we move all polygons up such that they still
	fit into the knapsack under $(1+\delta)$-resource augmentation and
	intuitively, each polygon has $\frac{\delta}{n}N$ units of empty
	space above and below it. For each polygon $P_{i}$ let $v_{i}$ be
	its first vertex $(x'_{i,1},y'_{i,1})$. We move each polygon $\tilde{P}_{i}$
	such that $v_{i}$ is placed on a point whose coordinates are integral
	multiples of $\frac{\delta}{4n}N$. For achieving this, it suffices
	to move $\tilde{P}_{i}$ by at most $\frac{\delta}{4n}N$ units down
	and by at most $\frac{\delta}{4n}N$ units to the left. By the above,
	we can do this for all polygons simultaneously without making them
	intersect. Also, intuitively each polygon still has $\frac{\delta}{4n}N$
	units of empty space around it in all four directions.
	
	Now we argue that we can rotate each polygon $\tilde{P}_{i}$ such
	that its angle is one out of $(n/\delta)^{O(1)}$ many possible angles.
	Consider a polygon $\tilde{P}_{i}$. Suppose that we rotate it around
	its vertex $v_{i}$. We want to argue that if we rotate $\tilde{P}_{i}$
	by an angle of at most $\frac{\delta}{16n}$ then this moves each
	vertex of $\tilde{P}_{i}$ by at most $\frac{\delta}{4n}N$ units.
	Let $v'_{i}$ be a vertex of $\tilde{P}_{i}$ with $v_{i}\ne v'_{i}$.
	Let $x$ denote the distance between the old and the new position
	of $v'_{i}$ if we rotate $\tilde{P}_{i}$ by an angle of $\alpha$.
	Then we have that $x=\sin\alpha\frac{\overline{v_{i}v'_{i}}}{\sin((\pi-\alpha)/2)}\le4N\alpha\le\frac{\delta}{4n}N$,
	assuming that $\alpha\le\frac{\delta}{16n}$ and that $\delta$ is
	sufficiently small.
	
	Therefore, we rotate $\tilde{P}_{i}$ around $v_{i}$ by an angle
	of at most $\frac{\delta}{16n}$ such that $d+\rot_{\alpha}(P_{i})=\tilde{P}_{i}$
	for an angle $\alpha$ which is an integral multiple of $\frac{\delta}{16n}$.
	Due to our movement of $\tilde{P}_{i}$ before we can assume that
	$d=\binom{d_1}{d_2}$ satisfies that $d_{1}$ and $d_{2}$ are
	integral multiples of $\frac{\delta}{4n}N$. Thus, for $d$ and for
	$\alpha$ there are only $(n/\delta)^{O(1)}$ possibilities which
	yields only $(n/\delta)^{O(1)}$ possible placements for $P_{i}$.
\end{proof}

\subsection{Recursive algorithm}

We describe our main algorithm. First, we guess how many polygons
of each of the shapes in $\S$ are contained in $\OPT$. Since there
are only $(\log(n)/\delta)^{O_{\delta}(1)}$ different shapes in $\S$
we can do this in time $n^{(\log(n)/\delta)^{O_{\delta}(1)}}$. Once
we know how many polygons of each shape we need to select, it is clear
which polygons we should take since if for some shape $S_{i}\in\S$
we need to select $n_{i}$ polygons with that shape then it is optimal
to select the $n_{i}$ polygons in $\P$ of shape $S_{i}$ with largest
weight. Therefore, in the sequel assume that we are given a set of
polygons $\P'\subseteq\P$ and we want to find a packing for them
inside $K$.

Our algorithm is recursive and it generalizes a similar algorithm
for the special case of axis-parallel rectangles in~\cite{adamaszek2014qptas}.
On a high level, we guess a partition of $K$ given by a separator
$\Gamma$ which is a polygon inside $K$. It has the property that
at most $\frac{2}{3}|\OPT|$ of the polygons of $\OPT$ lie inside
$\Gamma$ and at most $\frac{2}{3}|\OPT|$ of the polygons of $\OPT$
lie outside $\Gamma$. We guess how many polygons of each shape are
placed inside and outside $\Gamma$ in $\OPT$. Then we recurse separately
inside and outside $\Gamma$. For our partition, we are looking for
a polygon $\Gamma$ according to the following definition. 
\begin{defn}
Let $\ell\in\mathbb{N}$ and $\epsilon>0$. Let $\bar{\P}$ be a set
of pairwise non-overlapping polygons in $K$. A polygon $\Gamma$
is a \emph{balanced $\hat{\epsilon}$-cheap $\ell$-cut} if 
\end{defn}
\begin{itemize}
\item $\Gamma$ has at most $\ell$ edges, 
\item the polygons contained in $\Gamma$ have a total area of at most $2/3\cdot\area(\bar{\P})$, 
\item the polygons contained in the complement of $\Gamma$, i.e., in $K\setminus\Gamma$,
have a total area of at most $2/3\cdot\area(\bar{\P})$, and 
\item the polygons intersecting the boundary of $\Gamma$ have a total area
of at most $\hat{\epsilon}\cdot\area(\bar{\P})$.
\end{itemize}
In order to restrict the set of balanced cheap cuts to consider, we
will allow only polygons $\Gamma$ such that each of its vertices
is contained in a set $Q$ of size $(n/\delta)^{O(1)}$ defined as
follows. We fix a triangulation for each placement $P'_{i}\in\L_{i}$
of each polygon $P_{i}\in\P'$. We define a set $Q_{0}$ where for
each placement $P'_{i}\in\L_{i}$ for $P_{i}$ we add to $Q_{0}$
the positions of the vertices of $P'_{i}$. Also, we add the four
corners of $K$ to $Q_{0}$. Let $\V$ denote the set of vertical
lines $\{(\bar{x},\bar{y})|\bar{y}\in\R\}$ such that $\bar{x}$ is
the $x$-coordinate of one point in $Q_{0}$. We define a set $Q_{1}$
where for each placement $P'_{i}\in\L_{i}$ of each $P_{i}\in\P'$,
each edge $e$ of a triangle in the triangulation of $P'_{i}$, and
each vertical line $L\in\V$ we add to $Q_{1}$ the intersection of
$e$ and $L$. Also, we add to $Q_{1}$ the intersection of each line
in $L\in\V$ with the two boundary edges of $K$. Let $Q_{2}$ denote
the set of all intersections of pairs of line segments whose respective
endpoints are in $Q_{0}\cup Q_{1}$. We define $Q:=Q_{0}\cup Q_{1}\cup Q_{2}$.
A result in~\cite{adamaszek2014qptas} implies that there exists
a balanced cheap cut whose vertices are all contained in $Q$. 
\begin{lem}
[\cite{adamaszek2014qptas}]\label{lem:cheap-cut}Let $\epsilon>0$
and let $\P'$ be a set of pairwise non-intersecting polygons in the plane
with at most $\varLambda$ edges each such that $\area(P)<\area(\P')/3$ for each $P\in\P$. Then there
exists a balanced $O(\epsilon\varLambda)$-cheap $\varLambda(\frac{1}{\epsilon})^{O(1)}$-cut
$\Gamma$ whose vertices are contained in $Q$. 
\end{lem}
Our algorithm is recursive and places polygons from $\P'$, trying
to maximize the total area of the placed polygons. 
In each recursive call we are given an area $\bar{K}\subseteq K$
and a set of polygons $\bar{\P}\subseteq\P'$. In the main call these
parameters are $\bar{K}=K$ and $\bar{\P}=\P'$. If $\bar{\P}=\emptyset$
then we return an empty solution. If there is a polygon $P_{i}\in\bar{\P}$
with $\area(P_{i})\ge\area(\bar{\P})/3$ then we guess a placement
$P'_{i}\in\L_{i}$ and we recurse on the area $K\setminus P_{i}'$
and on the set $\bar{\P}\setminus\{P_{i}\}$. Otherwise, we guess
the balanced cheap cut $\Gamma$ due to Lemma~\ref{lem:cheap-cut}
with $\epsilon:=\frac{\delta}{\varLambda\log(n/\delta)}$ and for
each shape $S\in\S$ we guess how many polygons of $\P'$ with shape
$S$ are contained in $\Gamma\cap\bar{K}$, how many are contained
in $\bar{K}\setminus\Gamma$, and how many cross the boundary of $\Gamma$
(i.e., have non-empty intersection with the boundary of $\Gamma$).
Note that there are only $n^{(\Lambda\log(n/\delta))^{O(1)}}$ possibilities
to enumerate. Let $\bar{\P}_{\mathrm{in}},\bar{\P}{}_{\mathrm{out}}$,
and $\bar{\P}{}_{\mathrm{cross}}$ denote the respective sets of polygons.
Then we recurse on the area $\Gamma\cap\bar{K}$ with input polygons
$\bar{\P}_{\mathrm{in}}$ and on the area $\bar{K}\setminus\Gamma$
with input polygons $\bar{\P}{}_{\mathrm{out}}$. Suppose that the
recursive calls return two sets of polygons $\bar{\P}'_{\mathrm{in}}\subseteq\bar{\P}{}_{\mathrm{in}}$
and $\bar{\P}'{}_{\mathrm{out}}\subseteq\bar{\P}{}_{\mathrm{out}}$
that the algorithm managed to place inside the respective areas $\Gamma\cap\bar{K}$
and $\bar{K}\setminus\Gamma$. Then we return the set $\bar{\P}'_{\mathrm{in}}\cup\bar{\P}'_{\mathrm{out}}$
for the guesses of $\Gamma$, $\bar{\P}{}_{\mathrm{in}},\bar{\P}{}_{\mathrm{out}}$,
and $\bar{\P}{}_{\mathrm{cross}}$ that maximize $\area(\bar{\P}'_{\mathrm{in}}\cup\bar{\P}'_{\mathrm{out}})$.
If we guess the (correct) balanced cheap cut due to Lemma~\ref{lem:cheap-cut}
in each iteration then our recursion depth is $O(\log_{3/2}(n^{2}/\delta^{2}))=O(\log(n/\delta))$
since the cuts are balanced and each polygon has an area of at least
$\Omega(\area(K)\delta^{2}/n^{2})$ (see Lemma~\ref{lem:large-height-width}).
Therefore, if in a recursive call of the algorithm the recursion depth
is larger than $O(\log(n/\delta))$ then we return the empty set and
do not recurse further. Also, if we guess the correct cut in each
node of the recursion tree then we cut polygons whose total area is
at most a $\frac{\delta}{\log(n/\delta)}$-fraction of the area of
all remaining polygons. Since our recursion depth is $O(\log(n/\delta))$,
our algorithm outputs a packing for a set of polygons in $\P'$ with
area at least $(1-\frac{\delta}{\log(n/\delta)})^{O(\log(n/\delta))}\bar{w}(\bar{\P})=(1-O(\delta))\area(\bar{\P})$.
This implies the following lemma. 
\begin{lem}
Assume that there is a non-overlapping packing for $\P'$ in $K$.
There is an algorithm with a running time of $n^{(\Lambda\log (n/\delta))^{O(1)}}$
that computes a placement of a set of polygons $\bar{\P}'\subseteq\P'$
inside $K$ such that $\area(\bar{\P}')\ge(1-O(\delta))\area(\P')$. 
\end{lem}
It remains to pack the polygons in $\tilde{\P}':=\P'\setminus\bar{\P}'$.
The total area of their bounding boxes is bounded by $\sum_{P_{i}\in\tilde{\P}'}B_{i}\le2\area(\tilde{\P}')\le O(\delta)\area(\P')\le O(\delta)\area(K)$.
Therefore, we can pack them into additional space that we gain via
increasing the size of $K$ by a factor $1+O(\delta)$, using the Next-Fit-Decreasing-Height algorithm~\cite{coffman1980performance}. 
\begin{thm}
There is an algorithm with a running time of $n^{(\log (n)/\delta))^{O(1)}}$
that computes a set $\P'$ with $w(\P')\ge\OPT$ such that $\P'$
fits into $K$ under $(1+\delta)$-resource augmentation. 
\end{thm}

\bibliographystyle{alpha}
\bibliography{citations}

\end{document}